\documentclass[11pt,letterpaper]{article}

\usepackage{microtype}
\usepackage{graphicx}
\usepackage{subfigure}
\usepackage{booktabs}
\usepackage{caption}
\usepackage{fullpage}
\usepackage{hyperref}
\usepackage{amsmath}
\usepackage{amssymb}
\usepackage{amsthm}
\usepackage[algo2e,ruled,vlined]{algorithm2e}
\usepackage{hyperref}
\usepackage{xcolor}

\SetCommentSty{mycommfont}

\def\E{\textsf{E}}
\def\Exp{\textsf{Exp}}
\def\Var{\textsf{Var}}

\def\Zipf{\textsf{Zipf}}
\def\subZipf{\textsf{subZipf}}

\newtheorem{theorem}{Theorem}[section]
\newtheorem{lemma}[theorem]{Lemma}
\newtheorem{definition}[theorem]{Definition}
\newtheorem{corollary}[theorem]{Corollary}
\theoremstyle{remark}
\newtheorem{remark}[theorem]{Remark}

\renewcommand{\epsilon}{\varepsilon}

\begin{document}

\title{Composable Sketches for Functions of Frequencies:\\ Beyond the Worst Case}

\author{
Edith Cohen\\
Google Research and
Tel Aviv University\\
\texttt{edith@cohenwang.com}
\and
Ofir Geri\thanks{Most of this work was done while interning at Google Research.}\\
Stanford University\\
\texttt{ofirgeri@cs.stanford.edu}
\and
Rasmus Pagh\\
Google Research,
BARC, and
IT University of Copenhagen\\
\texttt{pagh@itu.dk}
}

\date{}

\maketitle

\begin{abstract}
Recently there has been increased interest in using machine learning techniques to improve classical algorithms.
In this paper we study when it is possible to construct compact, composable sketches
 for weighted
sampling and statistics estimation according to functions of data
frequencies.  Such structures are now central components of
large-scale data analytics and machine learning pipelines. However, many common
functions, such as thresholds and
$p$th frequency moments with $p>2$, are known to require polynomial-size sketches in the worst case.  We explore performance beyond the
worst case under two different types of assumptions.  The first is
having access to noisy \emph{advice} on item frequencies. This
continues the line of work of Hsu et al.~(ICLR 2019), who assume
predictions are provided by a machine learning model.
 The second is providing guaranteed performance on a restricted class of
input frequency distributions that are better aligned with what is
observed in practice. This extends the work on heavy hitters under Zipfian distributions in a seminal paper of Charikar et al.~(ICALP 2002).
Surprisingly, we show analytically and empirically that ``in practice'' small polylogarithmic-size sketches provide
accuracy for ``hard'' functions.
\end{abstract}

\section{Introduction}

Composable sketches, also known as mergeable
summaries~\cite{AgarwalMergeable},  are data
structures that support summarizing large amounts of
distributed or streamed data
with small computational resources (time, communication, and space).
Such sketches support processing
additional data elements and merging sketches of multiple datasets to obtain a sketch of
the union of the datasets.
 This design is suitable for working with streaming data (by processing elements as they arrive) and distributed datasets, and allows us to parallelize computations over massive datasets.
Sketches are now a central part of managing large-scale data, with
application areas as varied as federated learning~\cite{pmlr-v54-mcmahan17a} and statistics collection
at network switches~\cite{LiuMVSB:sigcomm2016,LiuBEKBFS:Sigcomm2019}.

The datasets we consider consist of {\em elements} that are key-value pairs $(x,v)$ where
$v\geq 0$.
The frequency $w_x$ of a key $x$ is defined as the sum of the values of elements with that
key.  When the values of all elements are $1$, the frequency is simply the
number of  occurrences of a key in the dataset.  Examples of such datasets
include search queries, network traffic, user interactions, or training data
from many devices. These datasets are typically
distributed or streamed.

Given a dataset of this form, one is often interested in computing statistics that depend on the frequencies of keys. For example, the statistics of interest can be the number of keys with frequency greater than some constant (threshold functions), or the second frequency moment ($\sum_x{w_x^2}$), which can be used to estimate the skew of the data. Generally, we are interested in statistics of the form
\begin{equation} \label{genstats:eq}
  \sum_x L_x f(w_x)
\end{equation}
where $f$ is some function
applied to the frequencies of the keys and the coefficients $L_x$ are provided (for
example as a function of the features of the key $x$).
An important special case, popularized in the seminal
work of~\cite{AMS}, is computing the $f$-statistics of the data:
 $\|f(\boldsymbol{w})\|_1 = \sum_x f(w_x)$.

One way to compute statistics of the form \eqref{genstats:eq} is to compute a random sample of keys, and then use the sample to compute estimates for the statistics. In order to compute low-error estimates, the sampling has to be weighted in a way that depends on the target function $f$: each key $x$ is weighted by $f(w_x)$. Since the problem of computing a weighted sample is more general than computing $f$-statistics, our focus in this work will be on composable sketches for weighted sampling according to different functions of frequencies.

The tasks of sampling or statistics computation can always be performed by first computing a table
of key and frequency pairs $(x,w_x)$.  But this aggregation requires a data structure
of size (and in turn, communication or space) that grows linearly with the number of keys whereas
ideally we want the size to grow at most
polylogarithmically.
With such small sketches we can only hope for approximate results and
generally we see a trade-off between sketch size, which determines the storage or
communication needs of the computation, and accuracy.

When estimating statistics from samples, the accuracy depends on the sample size and on how much the sampling probabilities ``suit'' the
statistics we are estimating. In order to minimize the error, the sampling probability of each key $x$ should be (roughly) proportional to $f(w_x)$. This leads to a natural and extensively-studied question: for which functions $f$ can we design efficient sampling sketches?

The literature and practice are ripe with surprising successes for sketching,
including small (polylogarithmic size) sketch structures for estimating the number of
distinct elements~\cite{FlajoletMartin85,hyperloglog:2007} ($f(w) =
I_{w>0}$), frequency moments ($f(w) = w^p$) for
$p\in [0,2]$~\cite{AMS,indyk:stable}, and computing $\ell_p$
heavy hitters (for $p\leq 2$, where an $\ell_p$ $\varepsilon$-heavy
hitter is a key $x$ with $w_x^p \geq \varepsilon \|\boldsymbol{w}\|^p_p$) ~\cite{MisraGries:1982,ccf:icalp2002,MM:vldb2002,CormodeMuthu:2005,spacesaving:ICDT2005}.
(We use $I_\sigma$ to denote the indicator function that is 1 if the predicate $\sigma$ is true, and 0 otherwise.)
A variety of methods now support sampling via small sketches for rich classes of
functions of frequencies~\cite{CapSampling,mcgregor2016better,JayaramW:Focs2018,CohenGeri:NeurIPS2019}, including the moments $f(w)=w^p$ for $p \in [0,2]$ and the family of concave sublinear functions.

The flip side is that we know of lower bounds that limit the
performance of sketches using small space for some fundamental tasks~\cite{AMS}.
A full characterization of functions $f$ for which $f$-statistics can be estimated using polylogarithmic-size sketches was provided in~\cite{BravermanOstro:STOC2010}.
Examples of ``hard'' functions are
thresholds $f(w) = I_{w>T}$ (counting the number of keys with frequency above
a specified threshold value $T$), threshold weights $f(w) = w
I_{w>T}$, and high frequency moments
$f(w) = w^p$ with
$p>2$.   Estimating the $p$th frequency moment ($\sum_x w_x^p$) for $p>2$ is known to require
space ~$\Omega(n^{1-2/p})$~\cite{AMS,li2013tight}, where $n$ is the
number of keys.  These particular functions are important for downstream
tasks: threshold aggregates characterize the
distribution of frequencies, and high moment estimation is used in
the method of moments, graph applications~\cite{EdenRS:siamdm2019}, and for estimating the
cardinality of multi-way self-joins~\cite{alon2002tracking} (a $p$th
moment is used for estimating a $p$-way join).

\medskip

{\bf Beyond the worst case.}
Much of the discussion of sketching classified functions into ``easy'' and ``hard''. For example, there are known efficient methods for sampling according to $f(w)=w^p$ for $p \leq 2$, while for $p > 2$, even the easier task of computing the $p$th moment is known to require polynomial space.
However, the hard data distributions used to establish lower bounds for
some functions of frequency are arguably not very realistic.
Real data tends to follow nice distributions and is often (at least somewhat) predictable.
We study sketching where additional assumptions
allow us to circumvent these lower bounds while still providing
theoretical guarantees on the quality of the estimates.
We consider two distinct ways of going beyond the worst case:
1) access to {\em advice
models}, and 2) making natural assumptions on the frequency distribution of the dataset.

For the sampling sketches described in this paper, we use a notion
of {\em overhead} to capture the discrepancy between the sampling
probabilities used in the sketch and the ``ideal'' sampling probabilities
of weighted sampling according to a target function of frequency $f$.
An immensely powerful
property of using sampling to estimate statistics of the form~\eqref{genstats:eq}  is that the overhead translates
into a multiplicative increase in sample/sketch size, without compromising
the accuracy of the results (with respect to what an ideal ``benchmark''
weighted sample provides).  This property was used in different contexts in
prior work, e.g., \cite{FriezeKV:JACM2004,multiw:VLDB2009}, and
we show that it can be harnessed for our purposes as well.
For the task of estimating $f$-statistics, we use a tailored definition of overhead, that is smaller than the overhead for the more general statistics \eqref{genstats:eq}.

\medskip
{\bf Advice model.}
The advice model for sketching was recently proposed and
studied by Hsu et al.~\cite{hsuIKV:ICLR2019}. The advice takes the form of an
oracle that is able to identify whether a given key is a heavy hitter.
Such advice can be generated,
for example, by a machine learning model trained on past data.  The
use of the ``predictability'' of data to improve performance was
also demonstrated in~\cite{KraskaBCDP:sigmod2018,IndykVY:NeurIPS2019}.
A similar heavy hitter oracle was used in \cite{jiangLLRW:ICLR2020} to study additional problems in the streaming setting. For high frequency moments, they obtained sketch size $O(n^{1/2-1/p})$, a quadratic
improvement over the worst-case lower bound.

Here we propose a sketch for {\em sampling by advice}.  We assume an advice oracle that returns a noisy prediction of the frequency of each key. This type of advice oracle was used in the experimental section of \cite{hsuIKV:ICLR2019} in order to detect heavy hitters. We show that when the predicted $f(w_x)$ for keys with
above-average contributions $f(w_x)$ is approximately accurate within a factor $C$, our sample has overhead
$O(C)$.  That is, the uncertainty in the advice
translates to a factor $O(C)$ increase in the sketch size but does \emph{not} impact the
accuracy.

\medskip
{\bf  Frequency-function combinations.}
Typically, one designs sketch structures to provide
guarantees for a certain function $f$ and any set of  input frequencies
$\boldsymbol{w}$.  The performance of a sketch structure is then analyzed for a
\emph{worst-case} frequency distribution.
The analysis of the advice model
also assumes worst-case distributions (with the benefit that comes
from the advice).    We depart from this and study sketch performance for a
{\em combination} $(F,W,h)$ of a family of functions $F$, a family $W$ of frequency distributions,
and an overhead factor $h$. Specifically, we seek sampling sketches
that produce
weighted samples with overhead at most $h$ with respect to
$f(\boldsymbol{w})$ for {\em every} function $f\in F$ and frequency distribution
${\boldsymbol w}\in W$.
By limiting the set $W$ of input frequency distributions we are able
to provide performance guarantees for a wider set $F$ of functions of
frequency, including functions that are worst-case hard.  We
particularly seek combinations with frequency distributions $W$
that are typical in practice.
Another powerful property of the combination formulation is that
it provides multi-objective
guarantees with respect to a multiple functions of frequency $F$ using the same sketch~\cite{MultiObjective,CapSampling,LiuMVSB:sigcomm2016}.

The performance of sketch structures on ``natural'' distributions was
previously considered in a seminal paper by  Charikar et al.~\cite{ccf:icalp2002}.
The paper introduced the \emph{Count Sketch} structure for heavy hitter detection, where
an $\ell_2$ $\varepsilon$-heavy hitter is a key $x$ with $w_x^2 \geq
\varepsilon  \|\boldsymbol{w}\|^2_2$.  They
also show that for Zipf-distributed data with parameter $\alpha \geq
1/2$, a count sketch of size $O(k)$ can be used to find the $k$ heaviest keys (a
worst-case hard problem) and that an $\ell_1$ sample can only identify
the heaviest keys for Zipf parameter $\alpha \geq 1$.

We  significantly  extend these insights to a wider family of
frequency distributions
and to a surprisingly broad class of functions of frequencies.
In particular we show that all  high moments ($p>2$) are ``easy''  as
long as the frequency distribution has an $\ell_1$ or $\ell_2$
$\varepsilon$-heavy hitter.  In this case, an $\ell_1$ or $\ell_2$ sample with
overhead $1/\varepsilon$ can be used to estimate all high moments.   We also show that in
a sense this characterization is tight in that if we allow all frequencies, we meet the known lower bounds.
It is very common for datasets in practice to have a most frequent key that
is an $\ell_1$ or $\ell_2$ $\varepsilon$-heavy hitter. This holds in particular for Zipf or approximate Zipf distributions.

Moreover, we show that
Zipf frequency distributions have
small {\em universal} sketches that apply to {\em any}
monotone function of frequency (including thresholds and high
moments).
Zipf frequencies were previously considered in the advice model~\cite{AamandIV:arxiv2019}.
Interestingly, we show that for these distributions a single small
sketch is effective with all monotone functions of frequency,  even
without advice.
In these cases, universal sampling is achieved with
off-the-shelf polylogarithmic-size sketches such as $\ell_p$ samples for $p\leq
2$ and multi-objective concave-sublinear samples
\cite{CCD:sigmetrics12,CapSampling,mcgregor2016better,JayaramW:Focs2018}.

\medskip
{\bf  Empirical study.}
We complement our analysis with an empirical study on multiple real-world
datasets including datasets studied in prior work on advice models \cite{AOL,CAIDA,StackOverflow,UGR}. (Additional discussion of the datasets appears in Section~\ref{datasets:sec}.)
We apply sampling by advice, with advice based on models from prior work or direct use of frequencies from
past data.  We then estimate high frequency moments from the samples. We observe that
sampling-by-advice was effective on these datasets, yielding low error
with small sample size.  We also observed, however, that $\ell_2$ and
$\ell_1$ samplers were surprisingly accurate on these tasks, with $\ell_2$ samplers
generally outperforming sampling by advice.  The surprisingly good performance of these simple sampling schemes is suggested from our analysis.

We compute the overhead factors for some
off-the-shelf sampling sketches on multiple
real-world datasets with the objectives of
$\ell_p$ sampling ($p>2$) and universal sampling. We find these factors to be surprisingly
small.  For example, the measured overhead of using $\ell_2$ sampling  for the objective of $\ell_p$ sampling ($p \geq 2$) is in the range $[1.18,21]$.  For universal sampling, the observed overhead
is lower with $\ell_1$ and with multi-objective concave
sublinear samples than with $\ell_2$ sampling and is in $[93,773]$,  comparing very
favorably with the alternative of computing a full table.
Finally, we use sketches to estimate the
distribution of rank versus frequency, which is an important tool for optimizing performance across application domains (for network flows, files, jobs, or search queries).  We find that $\ell_1$ samples provide quality estimates, which is explained by our analytical results.

\medskip
{\bf  Organization.}
In Section~\ref{sec:preliminaries}, we present the preliminaries, including the definition of overhead and description of off-the-shelf sampling sketches that we use. Our study of the advice model is presented in Section~\ref{sec:advice}. Our study of frequency-function combinations, particularly in the context of $\ell_p$ sampling, is presented in Section~\ref{sec:f-f-comb}. Section~\ref{uni:sec} discusses universal samples. Our experimental study is presented throughout Sections~\ref{sec:advice} and \ref{sec:f-f-comb}. Additional experimental results are reported in Appendix~\ref{actualmore:sec}.

\section{Preliminaries}\label{sec:preliminaries}
We consider datasets where each data element is a (key, value)
pair. The keys belong to a universe denoted by $\mathcal{X}$ (e.g.,
the set of possible users or words), and each key may appear in more
than one element. The values are positive, and for each key $x \in
\mathcal{X}$, we define its \emph{frequency} $w_x \geq 0$ to be the
sum of values of all elements with key $x$. If there are no elements with key $x$, $w_x = 0$. The data elements may
appear as a stream or be stored in a distributed manner. We denote
the number of active keys (keys with frequency greater than $0$) by $n$.

We are interested in sketches that produce a weighted sample of keys according to
some function $f$ of their frequencies, which means that the weight of each key $x$ is $f(w_x)$. In turn, the sampling probability of key $x$ is
roughly proportional to $f(w_x)$. We denote the vector of the frequencies of all active keys by $\boldsymbol{w}$ (in any fixed order).
We use $f(\boldsymbol{w})$ as a shorthand for the vector of all values $f(w_x)$.

\paragraph{Estimates from a sample.}
The focus of this work is sampling schemes that produce a random subset $S \subseteq \mathcal{X}$ of the keys in the
dataset. Each active key $x$ is included in the sample $S$ with probability $p_x$ that depends on the frequency $w_x$.
From such a sample, we can compute for each key $x$ the \emph{inverse probability estimate} \cite{HT52} of $f(w_x)$ defined as
\[
\widehat{f(w_x)} =
\begin{cases}
\frac{f(w_x)}{p_x} & \text{ if } x \in S\\
0 & \text{ if } x \notin S
\end{cases} \enspace .
\]
These per-key estimates are \emph{unbiased} ($\E\left[\widehat{f(w_x)}\right] = f(w_x)$). They can be summed to obtain unbiased estimates of
the $f$-statistics $\sum_{x\in H}{f(w_x)}$ of a domain $H\subset\mathcal{X}$:
\[ \widehat{\sum_{x\in H}f(w_x)} :=
  \sum_{x\in H} \widehat{f(w_x)}  = \sum_{x \in
    S \cap H}{\widehat{f(w_x)}}\enspace . \]
The last equality follows because $\widehat{f(w_x)}=0$ for keys
not in the sample. Generally, we use $\widehat{a}$ to denote the estimator for a quantity $a$ (e.g., $\widehat{\sum_{x\in H}f(w_x)}$ defined above is the estimator for $\sum_{x\in H}{f(w_x)}$).
We can similarly get \emph{unbiased} estimates for other statistics that are linear in
$f(w_x)$, e.g., $\sum_{x \in \mathcal{X}}{L_xf(w_x)}$ (for coefficients $L_x$).

\paragraph{Bottom-$k$ samples.}
We briefly describe a type of samples that will appear in our algorithms and analysis. In a \emph{bottom-$k$ sample} \cite{bottomk07:ds}, we draw a random value (called seed) for each active key $x$. The distribution of the seed typically depends on the frequency of the key. Then, to obtain a sample of size $k$, we keep the $k$ keys with lowest seed values. Many sampling schemes can be implemented as bottom-$k$ samples, including our sketch in Section~\ref{sec:advice}, PPSWOR, and sampling sketches for concave sublinear functions (the last two are described further in Section~\ref{offtheshelf:sec}).

\subsection{Benchmark Variance Bounds}\label{sec:benchmark-bounds}
In this work, we design sampling sketches and use them to compute estimates for some function $f$ applied to the frequencies of keys ($f(w_x)$ for key $x$). We measure performance with respect to that of a ``benchmark''  weighted sampling
scheme where the weight of each key $x$ is $f(w_x)$. Recall that for ``hard''  functions $f$, these
schemes can not be implemented with small sketches.

For an output sample of size $k$, these benchmark schemes include (i)~{\em probability proportional to size} (PPS) with
replacement, where the samples consists of $k$ independent draws in which key $x$ is
selected with probability $f(w_x)/\|f(\boldsymbol{w})\|_1$, (ii)~PPS without
replacement (PPSWOR~\cite{Rosen1972:successive,Rosen1997a,bottomk:VLDB2008}), or (iii)~priority
sampling~\cite{Ohlsson_SPS:1998,DLT:jacm07}.
When using PPS with replacement, we get unbiased estimators $\widehat{f(w_x)}$ of $f(w_x)$ for all keys $x$, and the variance is
upper bounded by
  \begin{equation} \label{keyvarbound:eq}
 \Var[\widehat{f(w_x)}] \leq  \frac{1}{k} f(w_x) \|f(\boldsymbol{w}) \|_1.
\end{equation}
Similar bounds (where the factor $k$ in the denominator is replaced by $k-2$) can be derived for PPSWOR (see, e.g., \cite{CapSampling}) and priority sampling \cite{DLTVarAnalysis}.\footnote{The bound appears in some texts as $\frac{1}{k-2} f(w_x) \|f(\boldsymbol{w}) \|_1$ and as $\frac{1}{k-1} f(w_x) \|f(\boldsymbol{w}) \|_1$ in others. Specifically, in bottom-$k$ implementations of PPSWOR/priority sampling, we need to store another value (the inclusion threshold) in addition to the sampled keys. If the $k$ keys we store include the threshold (so the sample size is effectively $k-1$), the bound has $k-2$ in the denominator. If we store $k$ keys and the inclusion threshold is stored separately (so we store a total of $k+1$ keys), the bound has $k-1$ in the denominator.}
\begin{remark}
We use an upper bound on the variance as the benchmark (instead of the exact variance) for the following reason: Assume for simplicity that $k=1$, in which case, in PPS each key $x$ is sampled with probability $p_x = f(w_x)/\|f(\boldsymbol{w})\|_1$. When $p_x$ approaches $1$ for some key $x$ (that is, one key dominates the data), the variance $\Var[\widehat{f(w_x)}] = (f(w_x))^2\left(\frac{1}{p_x}-1\right)$ of the inverse-probability estimator approaches $0$. Recall that the variance of PPS is a benchmark we are trying to get close to using a different sampling scheme. Since the variance when using PPS is $0$, we cannot approximate it multiplicatively if we use a sampling scheme where $x$ is sampled with a probability that multiplicatively approximates $p_x$. However, when there is not just one key that dominates the data, that is, when we know that $p_x \leq 1 - \frac{1}{c}$ for all $x$, we get that $\Var[\widehat{f(w_x)}] \geq (f(w_x))^2\frac{1}{cp_x}$, so the bound on the variance is almost tight.
\end{remark}
Consequently, the variance of the sum
estimator for the $f$-statistics of a  domain $H$ is bounded by (due to nonpositive covariance shown in earlier works, e.g., \cite{CapSampling}):
\begin{equation} \label{domainubound:eq}
 \Var\left[\widehat{\sum_{x\in H}f(w_x)}\right] \leq \frac{1}{k} \sum_{x\in H} f(w_x) \|f(\boldsymbol{w}) \|_1 .
\end{equation}
The variance on the estimate of $\|f(\boldsymbol{w}) \|_1$ is bounded by
\[ \Var[\widehat{\| f(\boldsymbol{w})\|_1}] \leq \frac{1}{k} \|f(\boldsymbol{w}) \|_1^2 \enspace .\]
With these ``benchmark'' schemes, if we wish to get multiplicative error bound (normalized root mean squared
error) of $\varepsilon$ for estimating $\| f(\boldsymbol{w})\|_1$, we need
sample size $k = O(\varepsilon^{-2})$.  We note that the estimates are
also concentrated in the Chernoff sense~\cite{MultiObjective,DLT:jacm07}.

We refer to the probability vector
\[ p_x := \frac{f(w_x)}{\|f(\boldsymbol{w})\|_1}\]  as the {\em PPS
  sampling probabilities} for $f(\boldsymbol{w})$.
When $f(w)=w^p$ (for $p>0$) we refer to sampling with the respective
PPS probabilities as $\ell_p$ sampling.

\paragraph{Emulating a weighted sample.}
Let $\boldsymbol{p}$ be the base PPS probabilities for $f(\boldsymbol{w})$.
When we use a weighted sampling scheme with weights $\boldsymbol{q} \not= \boldsymbol{p}$
then the variance bound \eqref{domainubound:eq} does not apply.
We will say that weighted sampling according to $\boldsymbol{q}$
{\em emulates} weighted sampling according to $\boldsymbol{p}$ with
{\em overhead} $h$ if for all $k$ and for all domains $H \subseteq \mathcal{X}$, a sample of
size $k h$ provides the variance bound \eqref{domainubound:eq} (and
the respective concentration bounds).
\begin{lemma} \label{varfactor:lemma}
The overhead of emulating weighted sampling according to $\boldsymbol{p}$ using
weighted sampling according to $\boldsymbol{q}$ is at most
\[h(p,q) := \max_x \frac{p_x}{q_x} \enspace .\]
\end{lemma}
\begin{proof}
We first bound the variance
of $\widehat{f(w_x)}$ for a key $x$ when using weighted sampling by
$\boldsymbol{q}$.
 Consider a weighted sample of size $k$
according to base probabilities $q_x$.  The probability of $x$ being included in a with-replacement sample of size $k$ is $1-(1-q_x)^k$. Then, the variance of the inverse-probability estimator is
\begin{align*}
 \Var[\widehat{f(w_x)}] &= (\E[\widehat{f(w_x)}])^2 - (f(w_x))^2\\
 &= (f(w_x))^2\left(\frac{1}{1-(1-q_x)^k}-1\right)\\
 &= (f(w_x))^2 \cdot \frac{(1-q_x)^k}{1-(1-q_x)^k}\\
 &= (f(w_x))^2 \cdot \frac{e^{\ln(1-q_x)k}}{1-e^{\ln(1-q_x)k}}\\
 &\leq (f(w_x))^2 \cdot \frac{1}{k\ln\left(\frac{1}{1 - q_x}\right)} & \left[\frac{e^{-x}}{1-e^{-x}} \leq \frac{1}{x}\right]\\
 &\leq (f(w_x))^2 \cdot \frac{1}{k q_x} & \left[\ln\left(\frac{1}{1-x}\right) \geq x \text{ since }1-x \leq e^{-x}\right]\\
                               &= \frac{1}{k}(f(w_x))^2 \cdot \frac{1}{p_x} \cdot \frac{p_x}{q_x}\\
  &= \frac{1}{k} f(w_x) \|f(\boldsymbol{w}) \|_1 \cdot \frac{p_x}{q_x} \enspace .
\end{align*}
The upper bound on the variance for a domain $H$ is:
\begin{equation} \label{pqdomainubound:eq}
 \frac{1}{k} \left(\sum_{x\in H} f(w_x)\right) \|f(\boldsymbol{w}) \|_1 \max_{x\in H} \frac{p_x}{q_x}\enspace .
\end{equation}
For any $H$,
the variance bound \eqref{pqdomainubound:eq} is larger
by the benchmark bound \eqref{domainubound:eq} by at most a factor of $h(p,q)$. Hence, a sample of size $kh(p,q)$ according to $\boldsymbol{q}$ emulates a sample of size $k$ according to $\boldsymbol{p}$.
\end{proof}
We emphasize the fact that a sample using $\boldsymbol{q}$
(instead of $\boldsymbol{p}$) gets the same accuracy as sampling using $\boldsymbol{p}$, as long as we increase the sample size accordingly.

\begin{remark} \label{multaccum:rem}
Overhead bounds accumulate multiplicatively:  If
 sampling according to $\boldsymbol{q}$ emulates a sample by
 $\boldsymbol{p}$ and a sample by $\boldsymbol{q'}$ emulates
 a sample by $\boldsymbol{q}$, then a sample by
 $\boldsymbol{q'}$ emulates a sample by $\boldsymbol{p}$
 with overhead $h(q',p) \leq h(q',q) h(q,p)$.
\end{remark}

\begin{remark}[Tightness]
The emulation overhead can be interpreted as providing an upper bound over all possible estimation tasks that the emulated sample could be used for.
This definition of overhead is tight if we wish to transfer
guarantees for all subsets $H$:  Consider a subset that is a single
key $x$ and sample size $k$ such that $q_x \leq p_x \ll 1/k$. The variance when sampling according
to $\boldsymbol{q}$ is  $\approx (f(w_x))^2 /(k q_x) = \frac{1}{k} f(w_x) \frac{f(w_x)}{p_x} \frac{p_x}{q_x} = \frac{1}{k} f(w_x) \|f(\boldsymbol{w})\|_1 \frac{p_x}{q_x}$.  This is a factor of $p_x/q_x$ larger than the variance when sampling according to $p_x = f(w_x)/\|f(\boldsymbol{w})\|_1$, which is
$\approx (f(w_x))^2 /(k p_x) = \frac{1}{k}
f(w_x)\|f(\boldsymbol{w})\|_1$.
\end{remark}

 \paragraph{Overhead for estimating $f$-statistics.}
If we are only interested in estimates of the full $f$-statistics
  $\|f(\boldsymbol{w})\|_1$, the overhead reduces to the expected
  ratio $\E_{x\sim \boldsymbol{p}} [p_x/q_x]$ instead of the maximum ratio.
 \begin{corollary}\label{cor:estimation-overhead}
Let $\boldsymbol{p}$ be the base PPS probabilities for
$f(\boldsymbol{w})$.  Consider weighted sampling of size $k$ according to $\boldsymbol{q}$.  Then,
\[
\Var[\widehat{\|f(\boldsymbol{w})\|_1}] \leq \sum_x \frac{1}{k} f(w_x) \|f(\boldsymbol{w}) \|_1 \cdot \frac{p_x}{q_x} =
\frac{\|f(\boldsymbol{w})\|^2_1}{k} \sum_x p_x \cdot \frac{p_x}{q_x} =
\frac{\|f(\boldsymbol{w})\|^2_1}{k} \E_{x\sim
  \boldsymbol{p}}\left[\frac{p_x}{q_x}\right] .\]
\end{corollary}

\paragraph{Multi-objective emulation.}
For $h \geq 1$ and a function of frequency $f$, suppose there is a family $F$ of functions such that a
 weighted sample according to $f$ emulates a weighted sample for every
 $g\in F$ with overhead $h$.
A helpful closure property of such $F$ is the following.
  \begin{lemma} \cite{MultiObjective} \label{MOclosure:lemma}
The family $F$ of functions emulated by weighted sampling according to $f$ is closed under nonnegative linear combinations.  That is, if
$\{f_i\}\subset F$ and $a_i \geq 0$, then $\sum_i a_i f_i \in F$.
\end{lemma}

\subsection{Off-the-Shelf Composable Sampling Sketches} \label{offtheshelf:sec}
We describe known polylogarithmic-size sampling sketches that we use or refer to in this work. These sketches are designed to provide statistical guarantees on the accuracy (bounds on the variance) with respect to specific functions of frequencies and all frequency distributions.  The samples still provide unbiased estimates of statistics with respect to any function of frequency.
We study the estimates provided by these off-the-shelf sketches through the lens of combinations: instead of considering a particular function of frequency and all frequency distributions, we study the overhead of more general frequency-function combinations.

\begin{enumerate}
    \item [(i)] \emph{$\ell_1$ sampling without replacement.}  A PPSWOR
sketch~\cite{CCD:sigmetrics12}  (building on the aggregated
scheme~\cite{Rosen1997a, bottomk:VLDB2008} and related
schemes~\cite{GM:sigmod98,EV:ATAP02}) of  size $k$ performs perfect without replacement sampling
of $k$ keys according to the weights $w_x$. The sketch stores
$k$ keys or hashes of keys. In one pass over the data, we can pick the sample of $k$ keys. In order to compute the inverse-probability estimator, we need to know the exact frequencies of the sampled keys. These frequencies and the corresponding inclusion probabilities can be obtained by a second pass over the dataset.
Alternatively, in a single streaming pass we can
collect partial counts that can be used with appropriate tailored estimators~\cite{CCD:sigmetrics12}.
\item[(ii)] \emph{$\ell_2$ sampling (and generally $\ell_p$ sampling for $p\in [0,2]$) with replacement.}
There are multiple designs based on linear projections~\cite{MonemizadehWoodruff,indyk:stable,AndoniKO:FOCS11,mcgregor2016better}.  Currently the best
asymptotic space bound is $O(\log \delta^{-1} \log^2 n)$ for $0 < p < 2$ (and $O(\log \delta^{-1} \log^3 n)$ for $p = 2$) for a single
sampled key, where $1 - \delta$ is the probability of
success in producing the sample~\cite{JayaramW:Focs2018}.  A with-replacement sample of size
$k$ can then be obtained with a sketch with $O(k \log^2 n \log
\delta^{-1})$ bits.\footnote{These constructions assume that the keys are integers between $1$ and
$n$. Also, the dependence on $\delta$ improves with $k$ but we omit this for brevity.}  We note that on skewed distributions without-replacement sampling is significantly more effective for a fixed sample size.
Recent work \cite{CohenPW:NeurIPS2020} has provided sampling sketches that have size $\tilde{O}(k)$ and perform
without-replacement $\ell_2$ sampling for $p\leq 2$.
\item[(iii)] \emph{Sampling sketches for capping functions~\cite{CapSampling} and
% more generally
concave sublinear functions~\cite{CohenGeri:NeurIPS2019}}.
We will also consider a multi-objective sample
that emulates all concave sublinear
functions of frequencies with space overhead $O(\log n)$ \cite{CapSampling}.
\end{enumerate}

In our analysis and experiments we compute the overhead of using the above sketches with respect to certain frequencies and functions. Recall that for each target function of frequency, the overhead is computed with respect to the applicable base PPS probabilities $p_x = \frac{f(w_x)}{\|f(\boldsymbol{w})\|_1}$.
Consider a frequency vector $\boldsymbol{w}$ in non-increasing order ($w_{i} \geq
w_{i+1}$).
The base PPS sampling probabilities for $\ell_p$ sampling are simply
$w_i^p/\|\boldsymbol{w}\|^p_p$.
The base PPS probabilities for multi-objective
concave sublinear sampling are $p_i = p'_i/\|\boldsymbol{p'}\|_1$ where
$p'_i := \frac{w_i}{i w_i+ \sum_{j=i+1}^n w_j}$. These samples emulate
sampling for all concave-sublinear functions with overhead
$\|\boldsymbol{p'}\|_1$.

\subsection{Datasets} \label{datasets:sec}
In our experiments, we use the following datasets:
\begin{itemize}
    \item AOL \cite{AOL}: A log of search queries collected over three months in 2006. For each query, its frequency is the number of lines in which it appeared (over the entire 92 days).
    \item CAIDA \cite{CAIDA}: Anonymous passive traffic traces from CAIDA's equinix-chicago monitor. We use the data collected over one minute (2016/01/21 13:29:00 UTC), and count the number of packets for each tuple (source IP, destination IP, source port, destination port, protocol).
    \item Stack Overflow (SO) \cite{StackOverflow}: A temporal graph of interactions between users on the Stack Overflow website. For each node in the graph, we consider its weighted in degree (total number of responses received by that user) and its weighted out degree (total number of responses written by that user).
    \item UGR \cite{UGR}: Real traffic information collected from the network of a Spanish ISP (for network security studies). We consider only one week of traffic (May 2016 Week 2). For a pair of source and destination IP addresses, its frequency will be the number packets sent between these two addresses (only considering flow labeled as ``background'', not suspicious activity).
\end{itemize}

\section{The Advice Model}\label{sec:advice}

In this section, we assume that in addition to the input, we are provided with oracle access to an
``advice'' model.  When presented with a key $x$, the advice model returns a
 prediction $a_x$ for the total frequency of
$x$ in the data.   For simplicity, we assume that the advice model consistently returns the same prediction for all queries with the same key.
This model is similar to the model used in \cite{hsuIKV:ICLR2019}.

At a high level, our sampling
sketch takes size parameters $(k_h,k_p,k_u)$, maintains a set of at most $k_h+k_p+k_u$ keys, and collects the exact frequencies $w_x$ for these stored keys.
The primary component of the sketch is a weighted sample by advice of size $k_p$. Every time an element with key $x$ arrives, we query the advice model and get $a_x$. The weighted sample by advice is a sample of $k_p$ keys according to the weights $f(\boldsymbol{a})$.
Our sketch stores keys according to two additional criteria in order
to provide robustness to the prediction quality of the advice:
\begin{itemize}
\item
The top $k_h$ keys according to the advice model. This provides tolerance to inaccuracies in
the advice for these heaviest keys.  Since these keys are included with probability $1$, they will not contribute to the error.
\item
A uniform sample of $k_u$ keys.  This allows keys that are ``below
 average'' in their contribution to $\|f(\boldsymbol{w})\|_1$ to be represented appropriately in the sample, regardless of the accuracy of the advice.  This provides robustness to the accuracy of the advice
 on these very infrequent keys and ensures they are not undersampled.  Moreover, this ensures that all
 active keys ($w_x>0$), including those with potentially no advice ($a_x=0$), have a positive probability of being sampled. This is necessary for unbiased estimation.
\end{itemize}

We provide an unbiased estimator that smoothly combines the different sketch components and provides the following guarantees.
\begin{theorem} \label{sampleA:lemma}
   Suppose the advice model is such that for some $c_p,c_u \geq 0$ and $h\geq 0$,
   all keys $x$ that are active ($w_x>0$) and
   not in the $h$ largest advice
values of active keys
   ($a_x < \{a_y \mid w_y>0 \}_{(n-h+1)}$)
   satisfy
\[\frac{f(w_x)}{\|f(\boldsymbol{w})\|_1} \leq  \max\{
  c_p \frac{f(a_x)}{\| f(\boldsymbol{a})\|_1},  c_u \cdot \frac{1}{n} \}\ .\]
Then for all $k\geq 1$, the estimates from a sample-by-advice sketch with $(k_h,k_p,k_u) = (h, \lceil k c_p \rceil+2, \lceil k c_u\rceil+2)$
satisfy the variance bound~\eqref{domainubound:eq} for all domains $H$.
\end{theorem}
The implementation details appear below in Section~\ref{sec:advice-implementation} and the proof appears in Section~\ref{sec:advice-analysis}.

In particular,
if our advice is approximately accurate, say $f(w_x)
\leq f(a_x) \leq c_p \cdot f(w_x)$, the overhead when sampling by advice is $c_p$.
\begin{corollary}
Let $f$ be such that
$f(w_x) \leq f(a_x) \leq c_p f(w_x)$. Then for all $k\geq 1$, the estimates from a sample-by-advice sketch with
sample size
$(k_h,k_p,k_u) = (0,\lceil k c_p\rceil+2,0)$ satisfy the variance bound~\eqref{domainubound:eq} for all domains $H$.
\end{corollary}

\subsection{Implementation}\label{sec:advice-implementation}

The pseudocode for our sampling-by-advice sketch and computing the respective estimators $\widehat{f(w_x)}$ is provided in Algorithms~\ref{alg:sampling-based} and \ref{alg:sampling-based-est}.
Since the advice $a_x$ predicting the total frequency of a key is available with each occurrence of the key,  we can implement the
weighted sample by advice $f(a_x)$
using schemes designed for
aggregated data (a model where each key occurs once with its full weight) such as \cite{Cha82,Ohlsson_SPS:1998,DLT:jacm07,Rosen1997a, bottomk:VLDB2008,varopt_full:CDKLT10}.
Some care (e.g., using a
hash function to draw the random seeds) is needed because keys may occur in multiple data elements.  In the pseudocode, we show how use a bottom-$k$ implementation of either PPSWOR or priority sampling. Since the sampling sketches are bottom-$k$ samples, they also support merging (pseudocode not shown). The pseudocode also integrates the uniform and by-advice samples to avoid duplication (keys that qualify for both samples are stored once).

\begin{algorithm2e}\caption{Sample by Advice (Data Processing)}\label{alg:sampling-based}
\DontPrintSemicolon
\KwIn{A stream of updates, advice model $a$, $k_h,k_p,k_u$ (sample
  size for heaviest, by-advice, and uniform)}
\KwOut{A sampling sketch for $f(\boldsymbol{w})$}
{\bf Initialization:}\;
\Begin{
    Draw a hash function $h$ such that $h(x) \sim \mathcal{D}$ independently for all $x$\tcp*{$\mathcal{D} = \Exp(1)$ for PPSWOR, $\mathcal{D} =  U[0,1]$ for priority}
    Create empty sample $S := (S.h, S.pu)$ \tcp*{$S.h$ stores top $k_h$ keys by advice; $S.pu$ is a bottom-$k$ sketch for a combined weighted-by-advice and uniform sample}\;
}

{\bf Process an update $(x,\Delta)$ with prediction $a_x$:}\;
\Begin{
    \If(\tcp*[h]{$x$ is in (any component of) the sample}){$x \in S$}
    {$w_x \gets w_x+\Delta$ \tcp*{Update the count of $x$}}
    \Else{
      \If(\tcp*[h]{$x$ has a top-$k_h$ advice value}){$|S.h|< k_h$ {\bf or}  $a_x > \min_{y\in S.h} a.y$ }{$w_x
      \gets \Delta$ \;
      insert $(x,w_x)$ to
      $S.h$ \;
             \If{$|S.h|=k_h+1$}{Eject $y \gets \arg\min_{z\in S.h} a_z$ from $S.h$ \tcp*{eject key with smallest advice in $S.h$}
             Process $(y,w_y)$ by the by-advice$+$uniform sampling sketch $S.pu$\;
             }
       }
    \Else{process $(x,\Delta)$ by the by-advice$+$uniform sampling sketch $S.pu$}
}
}
    \tcp{Internal subroutine: process update $(y,\Delta)$ by by-advice$+$uniform sampling sketch $S.pu$:}
    \Begin{
      $r_y \gets \frac{h(y)}{f(a_y)}$ \tcp*{Compute by-advice seed value of key $y$}
      \If(\tcp*[h]{$y$ has one of $k_p$ smallest $r_y$ and qualifies for by-advice sample}) {$r_y < \{r_z \mid z\in S.pu\}_{(k_p)}$}{$w_y\gets \Delta$ \;
      Insert $(y, w_y)$ to $S.pu$}
      \Else{
      \If(\tcp*[h]{$y$ has one of $k_u$ smallest $h(y)$ and qualifies for uniform sample}){$h(y)  < \{h(z) \mid z\in S.pu\}_{(k_u)}$}
      {$w_y\gets \Delta$ \;
      Insert $(y, w_y)$ to $S.pu$ together with $r_y$ and $h(y)$}
      }
      \ForEach{$z\in S.pu$ such that
            $h(z)  > \{h(z') \mid z'\in S.pu\}_{(k_u)}$ \textbf{ and } $r_z > \{r_{z'} \mid z'\in S.pu\}_{(k_p)}$}{Remove $z$ from $S.pu$}
      }

\end{algorithm2e}
\begin{algorithm2e}\caption{Sample by Advice (Estimator Computation)}\label{alg:sampling-based-est}
\DontPrintSemicolon
\KwIn{A by-advice sampling sketch for $f(\boldsymbol{w})$ with parameters $k_h,k_p,k_u$}
\KwOut{Sparse representation of $f(\boldsymbol{w})$ and estimate of $\|f(\boldsymbol{w})\|_1$}
\ForEach{$x$ in $S_h$}{$\widehat{f(w_x)} \gets f(w_x)$}
\ForEach(\tcp*[h]{keys stored in uniform/by-advice samples}){$x\in S.pu$}
{
  $\tau_p \gets \{r_{z} \mid z \in S.pu\setminus\{x\}\}_{(k_p - 1)}$\tcp*{The $(k_p - 1)^{th}$ smallest seed of a key other than $x$}
  $\tau_u \gets \{h(z) \mid z\in S.pu\setminus\{x\} \}_{(k_u - 1)}$ \tcp*{The $(k_u - 1)^{th}$ smallest hash value of a key other than $x$}
\If(\tcp*[h]{key $x$ is strictly included in the uniform or by-advice samples}){ $h(x) < \tau_u \text{ {\bf or }} r_x < \tau_p$}
{
$p_x \gets \Pr_{h}[h(x) < \max\{f(a_x)\tau_p, \tau_u\}]$
\tcp*{For PPSWOR $1-e^{-\max\{\tau_u, f(a_x)\tau_p\}}$; For priority $\min\{\max\{\tau_u, f(a_x)\tau_p\},1\}$}
$\widehat{f(w_x)} \gets  \frac{f(w_x)}{p_x}$\;
}
}
Return $\{(x,\widehat{f(w_x)})\}$ for all $x$ assigned with $\widehat{f(w_x)}$ (sparse representation of $f(\boldsymbol{w})$); The sum of the assigned $\widehat{f(w_x)}$ as an estimate of $\|f(\boldsymbol{w})\|_1$.
\end{algorithm2e}

\subsection{Analysis}\label{sec:advice-analysis}
We first establish that the sampling sketch returns the exact frequencies for sampled keys.
\begin{lemma}
The frequency $w_x$ of each key $x$ in the final sample is accurate.
\end{lemma}
\begin{proof}
Consider first the sample of size  that is weighted by the advice $f(a_x)$ (without the other two components: the uniform sample and the keys with largest advice values). Note that if a key enters the sample when the first element with that key is processed and remains stored, its
count will be accurate (we account for all the updates involving
that key). Since the prediction $a_x$ is consistent (i.e., the prediction
$a_x$ is the same in all updates involving $x$), the seed of $x$ is
the same every time $x$ appears. For $x$ to not enter the sample on
its first occurrence or to be removed at any point, there must be other $k$ keys in the sample with
seed values lower than that of $x$. If such keys exist, $x$ is not in the final sample.
The argument for the $k_h$
top advice keys and for the $k_u$ uniformly-sampled keys is similar.
\end{proof}

\begin{proof}[Proof of Theorem~\ref{sampleA:lemma}]
Since the estimators have nonpositive covariance, it suffices to establish the upper bound
\[\Var[\widehat{f(w_x)}] \leq
f(w_x)\|f(\boldsymbol{w})\|_1
\max\{
  \frac{c_p}{k_p-2},
  \frac{c_u}{k_u-2}
  \}
  \]
for every key $x$.

A key $x$ with one of the top $h$ advised frequencies ($a_x$ values) is included with probability $1$ and has
$\Var[\widehat{f(w_x)}] = 0$. For those keys, the claim trivially holds.
Otherwise, recall that we assume that for some
$c_p, c_u \geq 0$,
\[\frac{f(w_x)}{\|f(\boldsymbol{w})\|_1} \leq  \max\{
  c_p \frac{f(a_x)}{\| f(\boldsymbol{a})\|_1},  c_u \cdot \frac{1}{n} \}\ .\]

Since the estimates are unbiased,
\[\Var[\widehat{f(w_x)}] = \E[\widehat{f(w_x)}^2] - (f(w_x))^2 = (f(w_x))^2 \E_{p_x}[(1/p_x -1)],\]
where $p_x$ is as computed by Algorithm~\ref{alg:sampling-based-est}.  Now note that $p_x = \max\{p'_x, p''_x\}$, where $p'_x$ is the probability $x$ is included in a by-advice sample of size $k_p$ and $p''_x$ is the probability it is included in a uniform sample of size $k_u$. Each of these samples is a bottom-$k$ sample but with different sampling weights: for each key $x$, the sampling weight in the by-advice sample is $f(a_x)$ and in the uniform sample is $1$.
We bound the variance for each of the two samples, and obtain that the variance in the combined sample is the minimum of the two bounds.

Generally, given a PPSWOR sample or in priority sampling (as in our algorithm) with sampling weights $\boldsymbol{v}$, we get the following variance bound on the estimate for $v_x$:
\[
\Var[\widehat{v_x}] \leq \frac{1}{k-2}v_x\| \boldsymbol{v}\|_1.
\]
If we compute an estimate for $f(w_x)$ instead, we get
\[
\Var[\widehat{f(w_x)}] = \Var\left[\frac{f(w_x)}{v_x} \cdot \widehat{v_x}\right] = \left(\frac{f(w_x)}{v_x}\right)^2 \Var[\widehat{v_x}] \leq \left(\frac{f(w_x)}{v_x}\right)^2 \frac{v_x\| \boldsymbol{v}\|_1}{k-2} = \frac{(f(w_x))^2}{k-2} \cdot \frac{\|\boldsymbol{v}\|_1}{v_x}\ .
\]
We use this bound for each of the two samples.

The uniform sample is equivalent to setting $v_x = 1$ for all keys. Using the bound above, the variance from a uniform sample of size $k_u$ is bounded by
$\frac{1}{k_u-2} n (f(w_x))^2$.
If $\frac{f(w_x)}{\| f(\boldsymbol{w})\|_1} \leq c_u \cdot \frac{1}{n}$, we substitute
$f(w_x) \leq \| f(\boldsymbol{w})\|_1 c_u \cdot \frac{1}{n}$ and obtain
\begin{align*}
\Var[\widehat{f(w_x)}] &\leq \frac{1}{k_u-2} n (f(w_x))^2\\
&\leq \frac{1}{k_u-2} n f(w_x) \|f(\boldsymbol{w})\|_1 c_u \cdot \frac{1}{n} \\
&= \frac{c_u}{k_u-2}f(w_x) \|f(\boldsymbol{w})\|_1\ .
\end{align*}
The weight $v_x$ of key $x$ in the by-advice sample of size $k_p$ is $f(a_x)$, so the variance is bounded by
$\frac{1}{k_p-2}   \|f(\boldsymbol{a})\|_1 \frac{f(w_x)^2}{f(a_x)}$.
If $\frac{f(w_x)}{\|f(\boldsymbol{w})\|_1} \leq c_p \frac{f(a_x)}{\|f(\boldsymbol{a})\|_1}$, we similarly substitute and obtain that
\[
\Var[\widehat{f(w_x)}] \leq \frac{c_p}{k_p-2}  f(w_x) \|f(\boldsymbol{w})\|_1\enspace .
\]
Together, since at least one of the two cases must hold, we get that
\[\Var[\widehat{f(w_x)}] \leq
f(w_x)\|f(\boldsymbol{w})\|_1
\max\{
  \frac{c_p}{k_p-2},
  \frac{c_u}{k_u-2}
  \} \enspace .
  \]
The bound~\eqref{domainubound:eq} follows from the choice of $(k_h, k_p, k_u)$ and the nonpositive covariance of the estimators for different keys.
\end{proof}

\subsection{Experiments}
We evaluate the effectiveness of ``sampling by advice''  for estimating the frequency moments $\|\boldsymbol{w}\|_p^p$ with $p=3,7,10$ on
datasets from~\cite{AOL,StackOverflow} (the datasets are described in Section~\ref{datasets:sec}).   We use the following advice models:
%from prior work~\cite{hsuIKV:ICLR2019} based on a machine learning algorithms applied to past data and advice based directly on frequencies in past data.
\begin{itemize}
    \item AOL~\cite{AOL}: We use the same predictions as in \cite{hsuIKV:ICLR2019}, which were the result training a deep learning model on the queries from the first five days. We use the prediction to estimate frequency moments on the queries from the 51st and 81st days (after removing duplicate queries from multiple clicks on results).
    \item Stack Overflow~\cite{StackOverflow}: We consider two six-month periods, 1/2013--6/2013 and 7/2013--12/2013. We estimate the in and out degree moments only on the data from 7/2013--12/2013, where the advice for each node is its exact degree observed in earlier data (the previous six-month period).
\end{itemize}

Some representative results are reported in Figure~\ref{advice:fig}. Additional results are provided in Appendix~\ref{actualmore:sec}.
The results reported for sampling by advice are with $k_h=0$ and two choices of balance between the by-advice and the uniform samples: $k_p=k_u$ and $k_u = 32$. (The by-advice sample was implemented using PPSWOR as in Algorithm~\ref{alg:sampling-based}.)   We also report the performance of PPSWOR ($\ell_1$
sampling without replacement), $\ell_2$ sampling (with and without
replacement), and the benchmark upper bound for ideal PPS sampling according to $f(w)=w^p$, which is $1/\sqrt{k}$.

\begin{figure}[t]
\centering
  \includegraphics[width=0.49\textwidth]{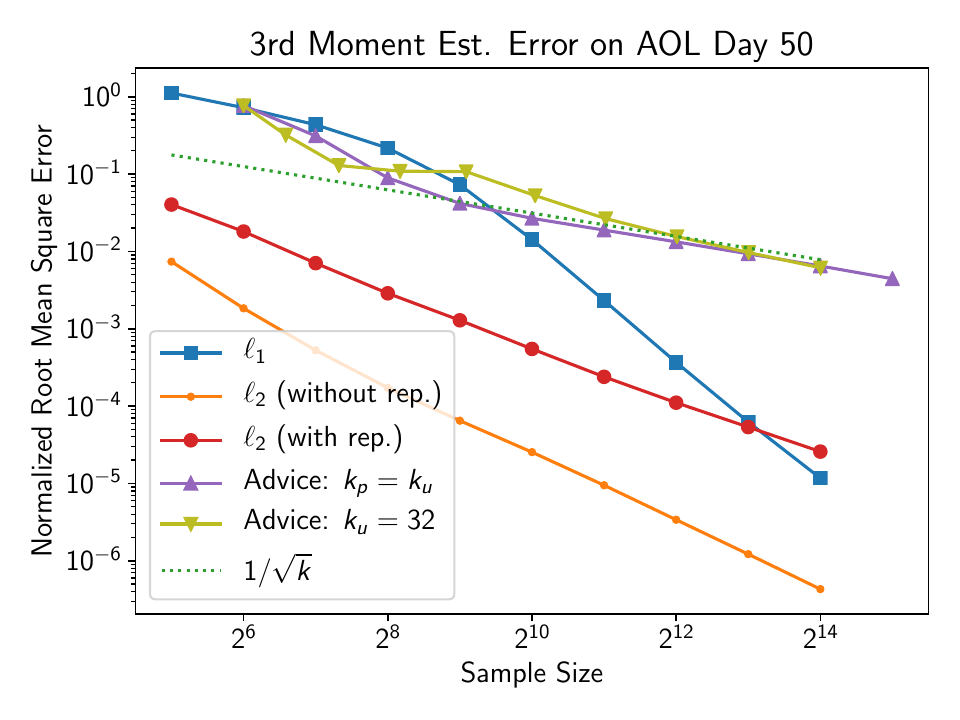}
  \includegraphics[width=0.49\textwidth]{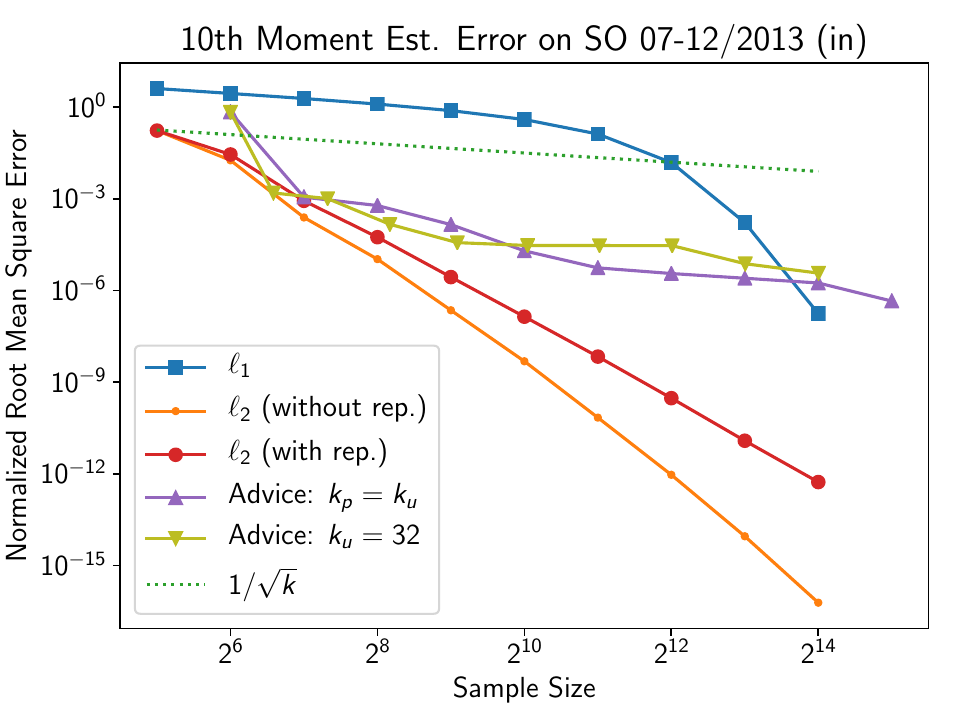}
\caption{Estimating the third moment on the AOL dataset (with learned
  advice) and the tenth moment on the Stack Overflow dataset (with past frequencies as advice).}
\label{advice:fig}
\end{figure}

\paragraph{Error estimation.} Our sampling scheme provides unbiased estimators for $\|f(\boldsymbol{w})\|_1$, where $f(w) = w^p$. We consider the normalized root mean square error (NRMSE), which for unbiased estimators is the same as the coefficient of variation (CV), the ratio of the standard deviation to the mean:
\[\frac{\left(\Var[\widehat{\|f(\boldsymbol{w})\|_1}]\right)^{1/2}}{\|f(\boldsymbol{w})\|_1} \enspace.\]
A simple way to estimate the variance is to use the
empirical squared error over multiple runs:  We take the average of $(\widehat{\|f(\boldsymbol{w})\|_1} - \|f(\boldsymbol{w})\|_1)^2$ over runs and apply a square root.  We found that 50 runs were not sufficient for an accurate estimate with our sample-by-advice methods.  This is because of keys that had relatively high frequency and low advice, which resulted in low inclusion probability and high contribution to the variance.  Samples that included these keys had higher estimates of the statistics than the bulk of other samples and often the significant increase was attributed to one or two keys.  This could be remedied by significantly increasing the number of runs we average over.
We instead opted to use different and more accurate estimators for the variance of the by-advice statistics and, for consistency, the baseline with and without-replacement schemes.

For with-replacement schemes we computed an upper bound on the variance (and hence the NRMSE) as follows.  The inclusion
   probability of a key $x$ in a sample of size $k$ is
   \[ p'_x := 1-(1-p_x)^k \] where
   $p_x$ is the probability to be selected in one sampling step. That is,  $p_x = w_x/\|\boldsymbol{w}\|_1$ for with-replacement $\ell_1$ sampling and $p_x = w_x^2/\|\boldsymbol{w}\|_2^2$ for with-replacement $\ell_2$ sampling.
   We can then compute the per-key variance of our estimator \emph{exactly} for each key $x$, which is
   $(1/p'_x-1)(f(w_x))^2$. Since estimates for different keys have nonpositive covariance, the variance of our estimate of the statistics $\|f(\boldsymbol{w})\|_1$
   is at most
   \[\Var[\widehat{\|f(\boldsymbol{w})\|_1}] \leq  \sum_x (1/p'_x-1)(f(w_x))^2 \enspace.\]

   For without-replacement schemes (sampling-by-advice and the without-replacement reference methods) we apply a more accurate estimator over the same set of $50$ runs.  For each ``run''  and each key $x$ (sampled or not) we consider all the possible samples where the randomization (i.e., seeds) of all keys $y\not= x$ is fixed as in the ``run.'' These include samples that include and do not include $x$.  We then compute the inclusion probability $p'_x$ of key $x$ under these conditions (fixed randomization for all other keys). The contribution to the variance due to this set of runs is
   $(1/p'_x-1)(f(w_x))^2$.
   We sum the estimates $(1/p'_x-1)(f(w_x))^2$ over all keys and take the average of the sums over runs as our estimate of the variance. The inclusion probability $p'_x$ is computed as follows.

   For bottom-$k$ sampling (in this case, PPSWOR by $w^q$ for $q=1,2$), recall that we compute random {\em seed} values to keys of the form $r_x/w^q_x$, where $r_x\sim \mathcal{D}$ are independent. The sample includes the $k-1$ keys with lowest seed values.   The inclusion probability is computed with respect to a threshold that is defined to be the $(k-1)$-th smallest ``seed'' value of other keys $\tau_x \gets \{\text{seed}(y) | y\not= x\}_{(k-1)}$.   For keys in the sample, this is the $k$-th smallest seed overall. For keys not in the sample the applicable $\tau_x$ is the $(k-1)$-th smallest seed overall.  We get
   $p'_x = \Pr_{r\sim \mathcal{D}}[r_x/w^q_x \leq \tau_x]$ (a different $p'_x$ is obtained in each run).

The calculation for the by-advice sampling sketch is as in the estimator in Algorithm~\ref{alg:sampling-based}, except that it is computed for all keys $x$.

\paragraph{Discussion.} We observe that for the third moment, sampling by advice did not perform significantly better than PPSWOR (and sometimes performed even worse). For higher moments, however, sampling by advice performed better than PPSWOR when the sample size is small.
  With-replacement $\ell_2$ sampling was more accurate than advice  (without-replacement $\ell_2$ sampling performs best).
  Our analysis in the next section explains the perhaps surprisingly good performance of $\ell_1$ and $\ell_2$ sampling schemes.

\section{Frequency-Function Combinations}\label{sec:f-f-comb}
In this section, we study the case of inputs that come from restricted families of
frequency distributions~$W$.  Restricting the family $W$ of possible
inputs allows us to extend the family $F$ of functions of frequencies
that can be efficiently emulated with small overhead.

Specifically,  we will consider sampling sketches and corresponding
{\em combinations} $(W,F,h)$  of frequency
vectors $W$, functions of frequencies $F$, and overhead $h \geq 1$, so
that
for every frequency distribution $w\in W$ and function $f\in F$, our
sampling sketch emulates
a weighted sample of
$f(\boldsymbol{w})$ with overhead $h$.  We will say that
our sketch {\em supports} the combination $(W,F,h)$.

Recall that emulation with overhead $h$ means that a sampling sketch
of size $h \epsilon^{-2}$ (i.e., holding that number of keys or hashes of keys)
provides estimates with
NRMSE $\epsilon$ for $\|f(\boldsymbol{w})\|_1$ for {\em any} $f\in F$ and
that these estimates are concentrated in a Chernoff-bound sense.
Moreover, for any $f\in F$ we can estimate statistics of the form $\sum_x L_x f(w_x)$
with the same guarantees on the accuracy provided by a dedicated weighted sample according to $f$.

We study combinations supported by the off-the-shelf sampling schemes that were described in Section~\ref{offtheshelf:sec}.
We will use the notation $\boldsymbol{w}=\{w_i\}$ for dataset
frequencies,  where $w_i$ is the frequency of the $i$-th most frequent key.

%We report results of experiments on datasets listed in
%Table~\ref{table:datasets} with details
%  provided in Appendix~\ref{datasets:sec}.~\nocite{AOL,UGR,StackOverflow,CAIDA}

\subsection[Emulating an l\_p Sample by an l\_q Sample]{Emulating an $\ell_p$ Sample by an $\ell_q$ Sample}
     We express the overhead of emulating $\ell_p$ sampling by
     $\ell_q$ sampling ($p \geq q$)
in terms of the properties of the frequency distribution.
Recall that $\ell_p$ sampling (and estimating the
       $p$-th
       frequency moment)
can be implemented with polylogarithmic-size sketches for $p\leq 2$
but requires polynomial-size sketches in the worst case when $p>2$.
               \begin{lemma} \label{genratio:lemma}
     Consider a dataset with frequencies $\boldsymbol{w}$ (in non-increasing order).
     For $p\geq q$, the overhead of emulating $\ell_p$ sampling by
     $\ell_q$ sampling is bounded by
     \begin{equation}\label{overheadpq:eq}
     \frac{\left\| \frac{\boldsymbol{w}}{w_1}\right\|^q_q}{\left\| \frac{\boldsymbol{w}}{w_1} \right\|^p_p }  .
    \end{equation}
     \end{lemma}
     \begin{proof}
     The sampling probabilities for key $i$ under $\ell_p$ sampling
     and $\ell_q$ sampling are $\frac{w_i^p}{\|w\|_p^p}$ and
     $\frac{w_i^q}{\|w\|_q^q}$, respectively. Then, the overhead of
     emulating $\ell_p$ sampling by $\ell_q$ sampling is
     \[
     \max_i{\tfrac{w_i^p/\|w\|_p^p}{w_i^q/\|w\|_q^q}} = \max_i{w_i^{p-q} \cdot \tfrac{\|w\|_q^q}{\|w\|_p^p}} = w_1^{p-q} \cdot \tfrac{\|w\|_q^q}{\|w\|_p^p} = \tfrac{\left\| \tfrac{\boldsymbol{w}}{w_1}\right\|^q_q}{\left\| \tfrac{\boldsymbol{w}}{w_1} \right\|^p_p }.
     \]
     \end{proof}
We can obtain a (weaker) upper bound on the overhead, expressed only
in terms of $q$, that applies to all $p\geq q$:
     \begin{corollary} \label{boundq:coro}
     The overhead of emulating $\ell_p$ sampling using $\ell_q$ sampling (for any $p\geq q$) is at most $\left\| \frac{\boldsymbol{w}}{w_1} \right\|^q_q$.
     \end{corollary}
     \begin{proof}
     For any set of frequencies $\boldsymbol{w}$, the normalized norm $\left\| \frac{\boldsymbol{w}}{w_1} \right\|^p_p$ is non-increasing with $p$ and is at least $1$.
     Therefore, the overhead \eqref{overheadpq:eq} is
           \[ \left\| \frac{\boldsymbol{w}}{w_1}\right\|^q_q \bigg/ \left\| \frac{\boldsymbol{w}}{w_1} \right\|^p_p \leq \left\| \frac{\boldsymbol{w}}{w_1} \right\|^q_q \enspace . \]
     \end{proof}
\begin{remark} Emulation as described above works when $p\geq q$.  When $q>p$, the maximum in
     the overhead bound (see proof of Lemma~\ref{genratio:lemma})  is incurred on the least frequent key, with
     frequency $w_n$.  We therefore get a bound of
     $\left\| \frac{\boldsymbol{w}}{w_n}\right\|^q_q / \left\|
         \frac{\boldsymbol{w}}{w_n} \right\|^p_p $ and Corollary~\ref{boundq:coro} does not apply.
\end{remark}

\subsection{Frequency Distributions with a Heavy Hitter}
 We show that for distributions with an $\ell_q$ heavy hitter,
 $\ell_q$ sampling emulates $\ell_p$ sampling for all $p\geq q$ with a
 small overhead.
\begin{definition}
 Consider  frequencies $\boldsymbol{w}$.
An $\ell_q$ $\phi$-heavy hitter is defined to be a key such that $w_i^q \geq
\phi \cdot \| \boldsymbol{w} \|^q_q$.\footnote{Another definition that is common in the literature uses $w_i \geq
\phi \cdot \| \boldsymbol{w} \|_q$ instead of $w_i^q \geq
\phi \cdot \| \boldsymbol{w} \|^q_q$.}
\end{definition}
We can now restate Corollary~\ref{boundq:coro}  in terms of  a presence of
a heavy hitter:
\begin{corollary} \label{hh:coro}
Let $\boldsymbol{w}$ be a frequency vector with a $\phi$-heavy
hitter under $\ell_q$. Then for $p\geq q$, the overhead of using
an $\ell_q$ sample to emulate an $\ell_p$ sample is at most $1/\phi$.
\end{corollary}
\begin{proof}
  If there is an $\ell_q$ $\phi$-heavy hitter, then the most frequent
  key (the key with frequency $w_1$) must be a $\phi$-heavy hitter.
From the definition of a heavy hitter, $\left\| \frac{\boldsymbol{w}}{w_1}\right\|^q_q \leq \frac{1}{\phi}$, and we get the desired bound on the overhead.
\end{proof}
We are now ready to specify combinations $(W,F,h)$  of frequency
vectors $W$, functions of frequencies $F$, and overhead $h \geq 1$ that
are supported by $\ell_q$ sampling.
\begin{theorem} \label{combinationsHH:thm}
  For any $q>0$ and $\phi\in (0,1]$, an $\ell_q$-sample supports
  the combination
  \begin{eqnarray}
W &:=& \{\text{$\boldsymbol{w}$  with an $\ell_q$  $\phi$-heavy hitter}\} \nonumber\\
F &:=& \overline{\{f(w)=w^p \mid p\geq q\}}_+ \nonumber\\
h &:=& 1/\phi    \enspace ,\nonumber
  \end{eqnarray}
 where the notation $\overline{F}_{+}$ is the
closure of a set $F$ of functions under nonnegative linear combinations.
 \end{theorem}
\begin{proof}
The claim for functions $f(w)=w^p$ is immediate from
Corollary~\ref{hh:coro}.  The claim for the nonnegative convex closure
of these function is a consequence of Lemma~\ref{MOclosure:lemma}.
\end{proof}
In particular, if the input distribution has an $\ell_q$  $\phi$-heavy hitter, then an $\ell_q$ sample of size
$\epsilon^{-2}/\phi$ emulates an $\ell_p$ sample of size
$\epsilon^{-2}$ for any $p>q$.

\begin{table}[tbh]
\caption{Datasets}
\label{table:datasets}
\centering
\begin{tabular}{|c|c|c|c|c|c|}
\hline
Dataset & $n/10^6$ & $\ln n$ & $\ell_1$ HH & $\ell_2$ HH & Zipf Parameter\\
\hline
\hline
SO out & $2.23$ & $14.6$ & $0.0016,0.0010$  & $0.0532,0.0221$ & $1.48$ \\
\hline
SO in & $2.30$ & $14.6$ & $0.0014,0.0008$  & $0.1068,0.0364$ & $1.38$  \\
\hline
AOL & $10.15$ & $16.1$ & $0.0274,0.0091$ & $0.8268,0.0911$ & $0.77$  \\
\hline
CAIDA & $1.07$ & $13.9$ & $0.0033,0.0032$ & $0.0479,0.0458$ & $1.35$  \\
\hline
UGR & $79.38$ & $18.2$ & $0.1117,0.0400$ & $0.8499,0.1093$ & $1.35$ \\
\hline
\end{tabular}
\end{table}
\begin{table}[tbh]
\caption{Overhead}\label{table:overhead}
\centering
{\scriptsize
\begin{tabular}{|c|c|c|c|c|c|c|c|}
\hline
Dataset & \multicolumn{3}{c|}{$\ell_1$ Emul.\ Overhead (Est.\ Overhead)} & \multicolumn{3}{c|}{$\ell_2$ Emul.\ Overhead (Est.\ Overhead)} & Concave\\
& $3$rd & $10$th & Universal & $3$rd & $10$th & Universal & Universal\\
\hline
SO out & 124.30 (42.76) & 600.44 (577.57) & 624.58 & 3.74 (1.90) & 18.04 (17.36) & $1.25 \times 10^5$ & 1672.60\\
\hline
SO in & 299.80 (155.32) & 677.56 (673.45) & 681.72 & 4.12 (2.58) & 9.30 (9.25) & $4.20 \times 10^4$ & 1628.36\\
\hline
AOL & 34.81 (33.45) & 36.38 (36.37) & 92.92 & 1.16 (1.12) & 1.21 (1.21) & $2.94 \times 10^5$ & 170.84\\
\hline
CAIDA & 31.23 (18.73) & 90.66 (56.20) & 301.28 & 2.16 (1.57) & 6.28 (4.03) & $2.65 \times 10^5$ & 846.15\\
\hline
UGR & 8.52 (8.17) & 8.95 (8.95) & 772.83 & 1.12 (1.09) & 1.18 (1.18) & $1.89 \times 10^{11}$ & 143.94\\
\hline
\end{tabular}
}
\end{table}

Table~\ref{table:datasets} reports properties and the relative $\ell_1$ and $\ell_2$ weights of the
two most frequent keys for our datasets (described in Section~\ref{datasets:sec}).
We can see that the most frequent key is a heavy hitter with
$1/\phi\leq 21$ for $\ell_2$ and $1/\phi \leq 715$ for $\ell_1$ which
gives us upper bounds on the overhead of emulating any $\ell_p$ sample
($p\geq 2$) by an $\ell_1$ or $\ell_2$ sample.
Table~\ref{table:overhead} reports (for $p=3,10$)
the overhead of emulating the respective $\ell_p$ sample and the (smaller)
overhead of estimating the $p$-th moment. (Recall the definitions of both types of overheads from Section~\ref{sec:benchmark-bounds}.)
We can see that high moments can be estimated well from $\ell_2$ samples and
with a larger overhead from $\ell_1$ samples.

\paragraph{Certified emulation.}
The quality guarantees of a combination $(W,F,h)$ are provided when
$\boldsymbol{w}\in W$.  In
practice, however, we may compute samples of arbitrary
dataset frequencies $\boldsymbol{w}$.   Conveniently, we are able to test
the validity of emulation by considering the most frequent
key in the sample:
For an $\ell_q$ sample of size $k$ we can compute
$r \gets \max_{x\in S} w_x^q/\|\boldsymbol{w}\|^q_q$ and certify that
our sample emulates $\ell_p$ samples $(p>q$) of size
$k r$.   If $kr$ is small, then we do not provide meaningful
accuracy but otherwise we can certify the emulation with sample size $kr$.
When the input  $\boldsymbol{w}$ has
an $\ell_q$  $\phi$-heavy hitter then a with-replacement
$\ell_q$ sample of size $k$ will include it with probability at least
$1-e^{-k \phi}$ and the result will be certified.  Note that the
result can only be certified if there is a heavy hitter.

\paragraph{Tradeoff between $W$ and $F$.}
If $\boldsymbol{w}$ has an $\ell_q$
$\phi$-heavy hitter, then $w_1$ is an $\ell_q$
$\phi$-heavy hitter, and is also an $\ell_p$  $\phi$-heavy hitter for every
$p\geq q$.    This means that for the target of moments with $p\geq 2$, an $\ell_2$
sample supports a larger set $W$ of frequencies than an
$\ell_1$ sample, including those with an $\ell_2$ $\phi$-heavy hitter
but not an $\ell_1$ $\phi$-heavy hitter.  The $\ell_1$ sample
however supports a larger family of target functions $F$ that includes
moments with $p\in [1,2)$.  Note that for a fixed overhead, with $\ell_q$ sampling the set $F$ of
supported functions decreases with $q$ whereas $W$ increases with $q$.

\subsection{Zipfian and Sub-Zipfian Frequencies}
Zipf distributions are very commonly used to model frequency
 distributions in practice.  We explore
supported combinations with frequencies that are (approximately)
Zipf.
\begin{definition}\label{def:zipf}
We say that a frequency vector $\boldsymbol{w}$ (of size $\|\boldsymbol{w}\|_0=n$) is $\Zipf[\alpha,n]$
  (Zipf with parameter $\alpha$) if for all $i$,  $w_i/w_1 = i^{-\alpha}$.
   \end{definition}
Values $\alpha\in  [1,2]$ are common in practice.
The best-fit Zipf parameter for the datasets we studied is reported in Table~\ref{table:datasets}, and their
 frequency distributions (sorted by rank) are shown in
 Figure~\ref{frequencies:fig}. We can see that our datasets are approximately Zipf (which would be an approximate straight line) and for all but one we have $\alpha\in [1.3,1.5]$.

\begin{figure}[tbh]
\centering
\includegraphics[width=0.49\linewidth]{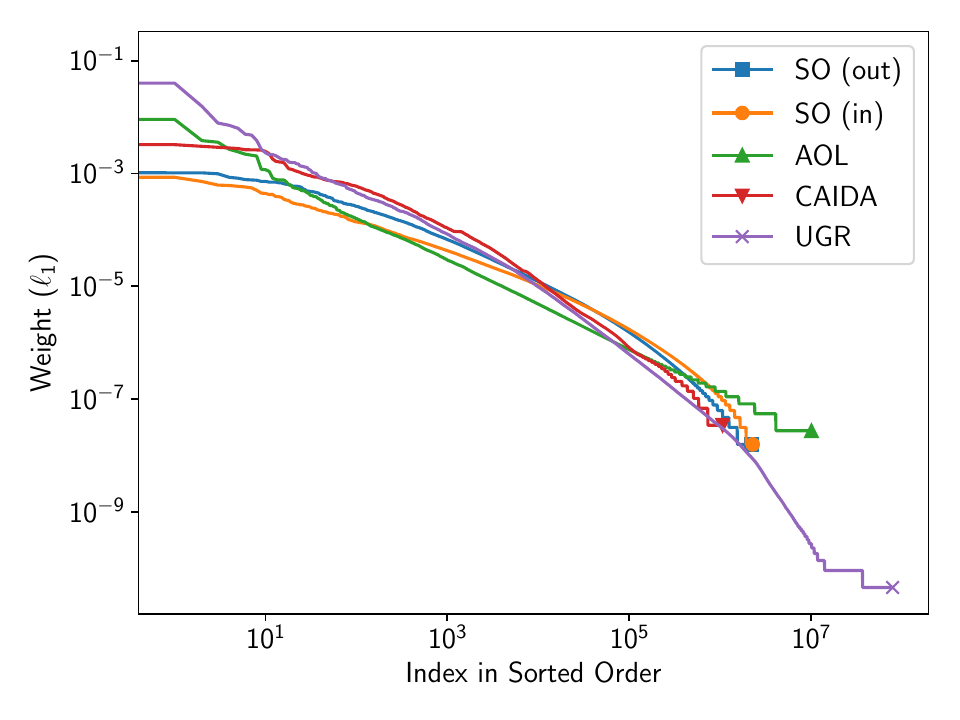}
\caption{The frequency distributions of the datasets (sorted in non-increasing order).}\label{frequencies:fig}
\end{figure}

We now define a broader class of approximately Zipfian distributions.
 \begin{definition}
A frequency vector $\boldsymbol{w}$ is $\subZipf[\alpha,c,n]$ if for all
     $i$, $\frac{w_i}{w_1}\leq ci^{-\alpha}$.\footnote{This is a slight abuse
     of notation since the parameters do not fully specify a
     distribution.}
 \end{definition}
Note that $\Zipf[\alpha,n]$ is sub-Zipfian with the same $\alpha$ and $c=1$.
 We show that sub-Zipf frequencies (and in particular Zipf
   frequencies) have heavy hitters.
   \begin{lemma} \label{zipfHH:lemma}
If a frequency vector $\boldsymbol{w}$ is $\subZipf[\alpha,c,n]$, and $q$ is such that $q\alpha \geq
   1$, the frequency vector has an $\ell_q$   $\frac{1}{c^qH_{n,\alpha
       q}}$-heavy hitter, where
 $H_{n,\alpha} := \sum_{i=1}^n i^{-\alpha}$ are the generalized
     harmonic numbers.
 \end{lemma}
 \begin{proof}
      By the definition of $\subZipf[\alpha,c,n]$ frequencies,
    \[
      \left\| \frac{\boldsymbol{w}}{w_1}\right\|^q_q = \sum_{i=1}^n \left(\frac{w_i}{w_1}\right)^q  \leq c^q
        \sum_{i=1}^n i^{-\alpha q} = c^q H_{n,\alpha q}\enspace .
    \]
    Hence, the largest element is an $\ell_q$   $\frac{1}{c^qH_{n,\alpha
       q}}$-heavy hitter.
  \end{proof}

\begin{table}\caption{Supported Combinations for
    $\subZipf[\alpha,c,n]$ Frequencies}\label{subzipf:table}
    \centering
          \begin{tabular}{l||l|l|l}
Method &     $\subZipf$ Parameters & $F$ & Overhead  \\
       \hline
 $\ell_1$ sampling &       $\{\alpha\geq 2\}$ & $\overline{\{ f(w)=w^p \mid p\geq 1\}}_+$ &
                                                               $1.65 c$ \\
 $\ell_1$ sampling &     $\{ \alpha\geq 1\}$ & $\overline{\{f(w)=w^p \mid p\geq
                                              1\} }_+$ & $ (1+\ln n) c$ \\
\hline
$\ell_2$ sampling &   $\{\alpha\geq 1\}$ & $\overline{\{f(w)=w^p \mid p\geq 2\}}_+$ &
                                                              $1.65 c^2$ \\

$\ell_2$ sampling &    $\{\alpha\geq 1/2\}$ & $\overline{\{f(w)=w^p \mid p\geq 2\}}_+$ &
                                                              $(1+\ln n) c^2$

          \end{tabular}
\end{table}
Table~\ref{subzipf:table} lists supported combinations that include these approximately Zipfian
distributions.
\begin{lemma}
       The combinations shown in Table~\ref{subzipf:table}
       are supported by $\ell_1$ and $\ell_2$ samples.
\end{lemma}
\begin{proof}
  We use Lemma~\ref{zipfHH:lemma} and Theorem~\ref{combinationsHH:thm}.
  Recall that when $\alpha=1$, the harmonic sum is $H_{n,1} \leq 1+\ln n$. For
 $\alpha>1$, $H_{n,\alpha}\leq \zeta(\alpha)$, where
$\zeta(\alpha) := \sum_{i=1}^\infty i^{-\alpha}$ is the Zeta function.
     The Zeta function is decreasing with $\alpha$, defined for  $\alpha>1$ with an asymptote at
     $\alpha=1$, and is at most $1.65$ for $\alpha\geq 2$.

           When $q\alpha \geq 2$, the overhead is at most $c^q \zeta(q\alpha) \leq
     1.65 c^q$. When $q\alpha = 1$ the overhead is at most $(1+\ln n) c^q$ and
     when $q\alpha>1$ we can bound it by $\min\{1+\ln n,
     \zeta(q\alpha)\} c^q$.
\end{proof}
We see that for these approximately Zipf distributions,
$\ell_1$ or $\ell_2$ samples emulate $\ell_p$ samples with small overhead.

\subsection{Experiments on Estimate Quality}
The overhead factors reported in Table~\ref{table:overhead} are in a sense worst-case upper bounds (for the dataset frequencies).
Figure~\ref{momentsest:fig} reports the actual estimation error (normalized root mean square error) for high moments for representative datasets as a function of sample size.  The estimates are with PPSWOR
($\ell_1$ sampling without replacement) and $\ell_2$ sampling with and
without replacement.
Additional results are reported in Appendix~\ref{actualmore:sec}.  We observe that the actual accuracy is significantly better than even the benchmark bounds.

Finally we consider estimating the full distribution of frequencies, that is, the curve that relates the frequencies of keys to their ranks. We do this by estimating the actual rank of each key in the sample (using an appropriate threshold function of frequency) --- more details are in Appendix~\ref{actualmore:sec}.  Representative results are reported in Figure~\ref{estfreq:fig} for PPSWOR and with-replacement $\ell_2$ sampling (additional results are reported in Appendix~\ref{actualmore:sec}). We used samples of size $k=32$ or $k=1024$ for each set of estimates.  We observe that generally the estimates are fairly accurate even with a small sample size (despite the fact that threshold functions require large sketches in the worst case). We see that $\ell_2$ samples are accurate for the frequent keys but often have no representatives from the tail whereas the without-replacement $\ell_1$ samples are more accurate on the tail.

\begin{figure}[thb]
\centering
\includegraphics[width=0.49\textwidth]{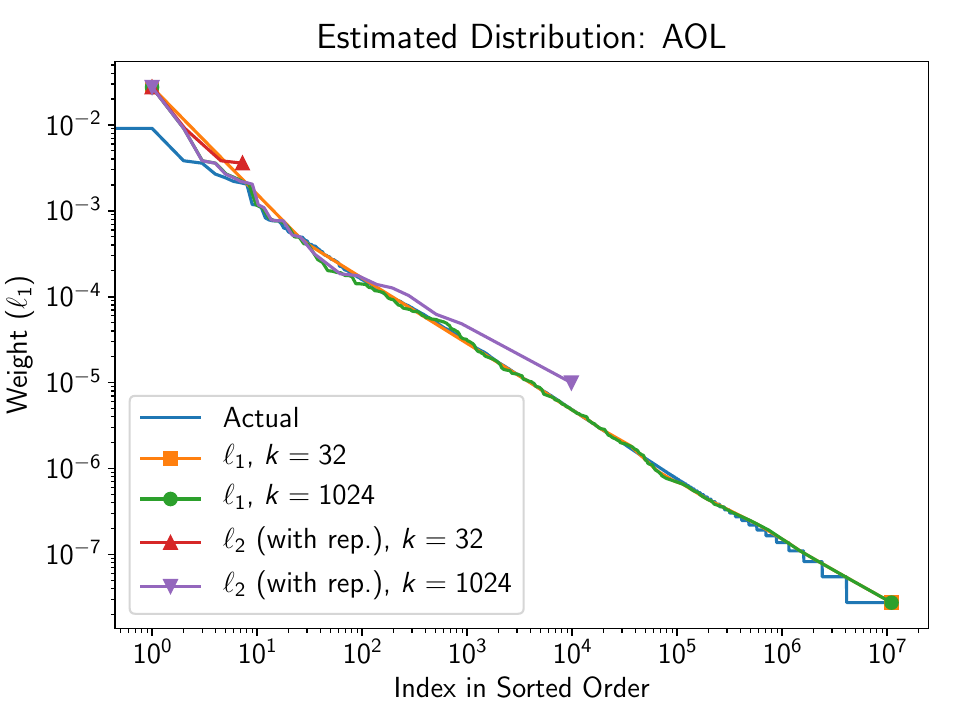}
\includegraphics[width=0.49\textwidth]{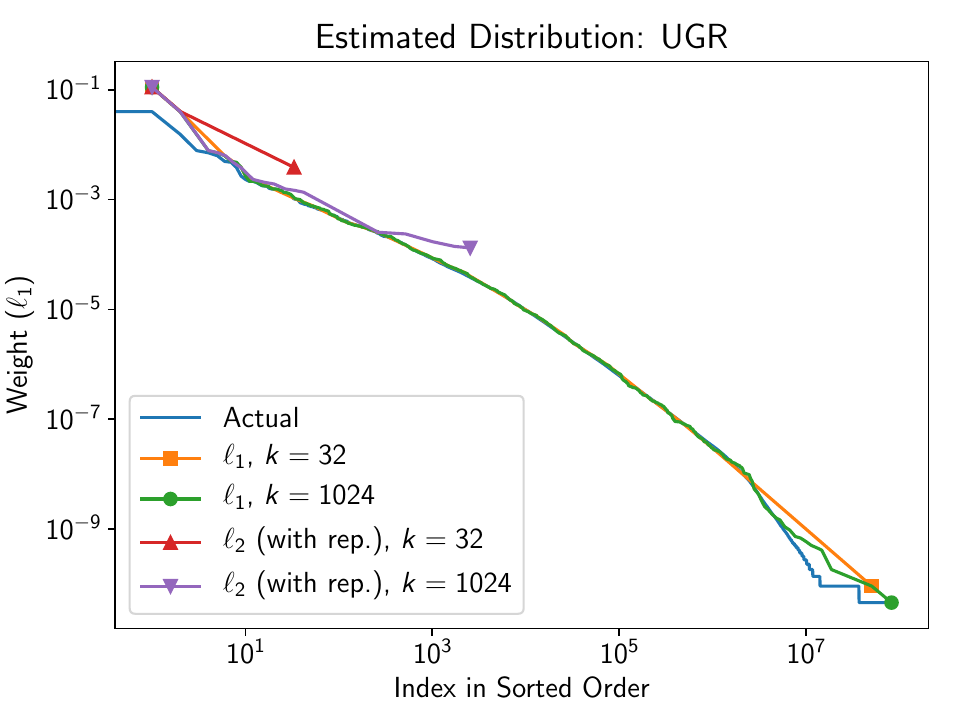}
\caption{Actual and estimated distribution of frequency by rank.  The estimates use PPSWOR and with-replacement $\ell_2$ sampling and sample sizes $k=32,1024$.\label{estfreq:fig}}
\end{figure}

\begin{figure}[thb]
\centering
    \includegraphics[width=0.49\textwidth]{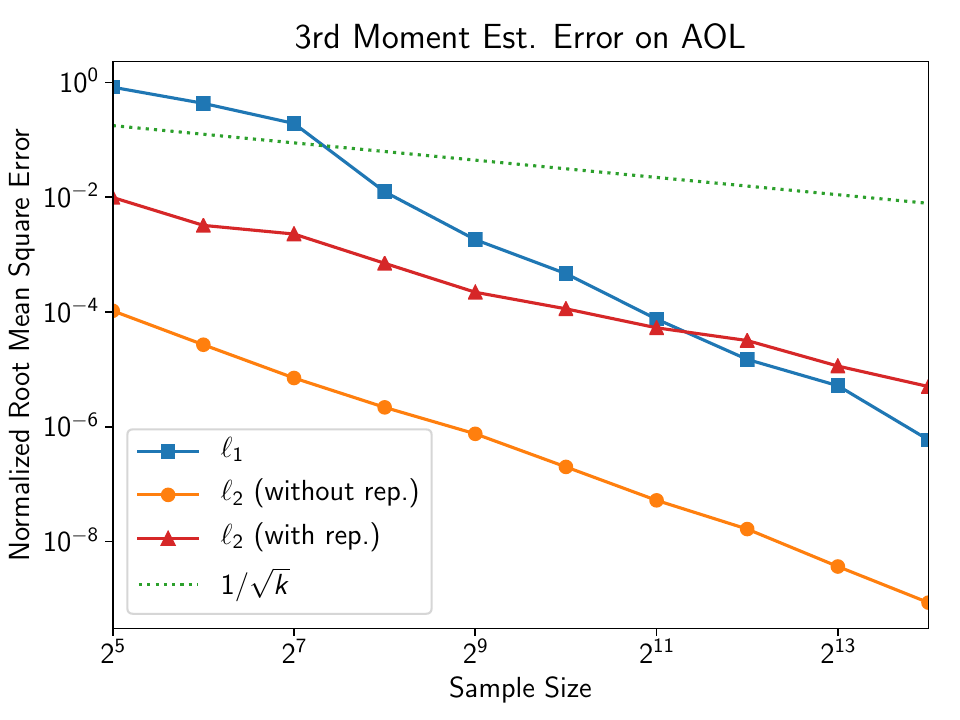}
    \includegraphics[width=0.49\textwidth]{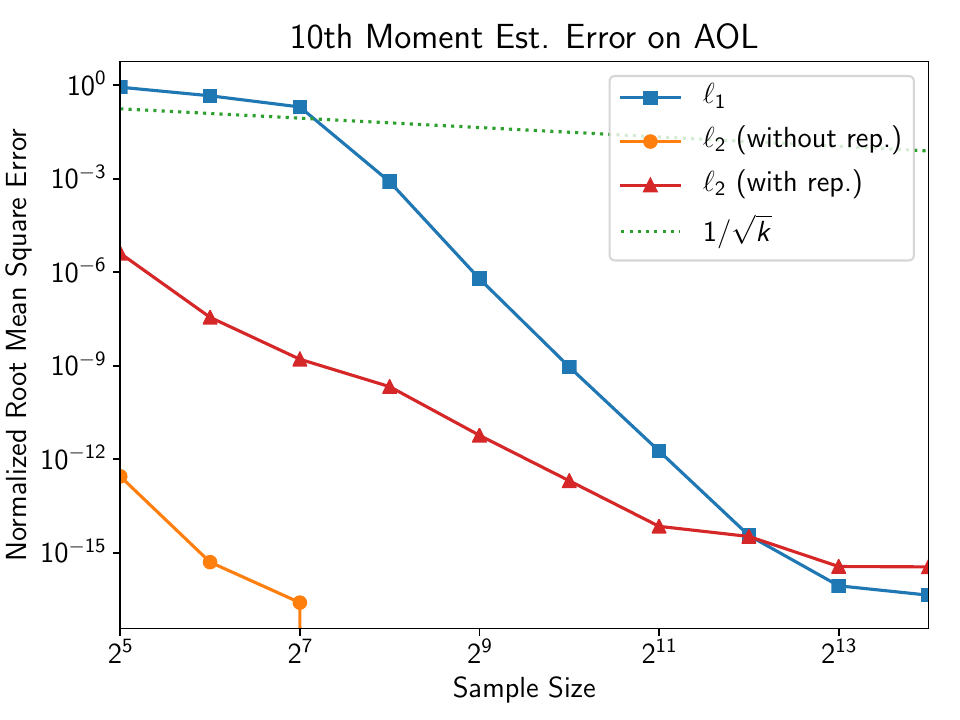}
    \includegraphics[width=0.49\textwidth]{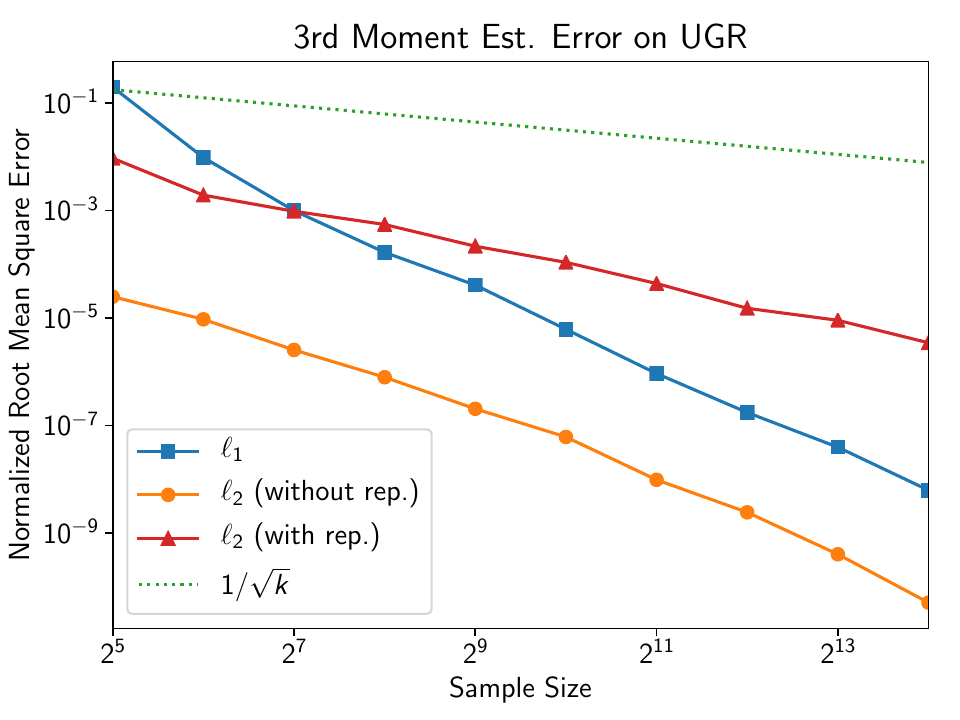}
    \includegraphics[width=0.49\textwidth]{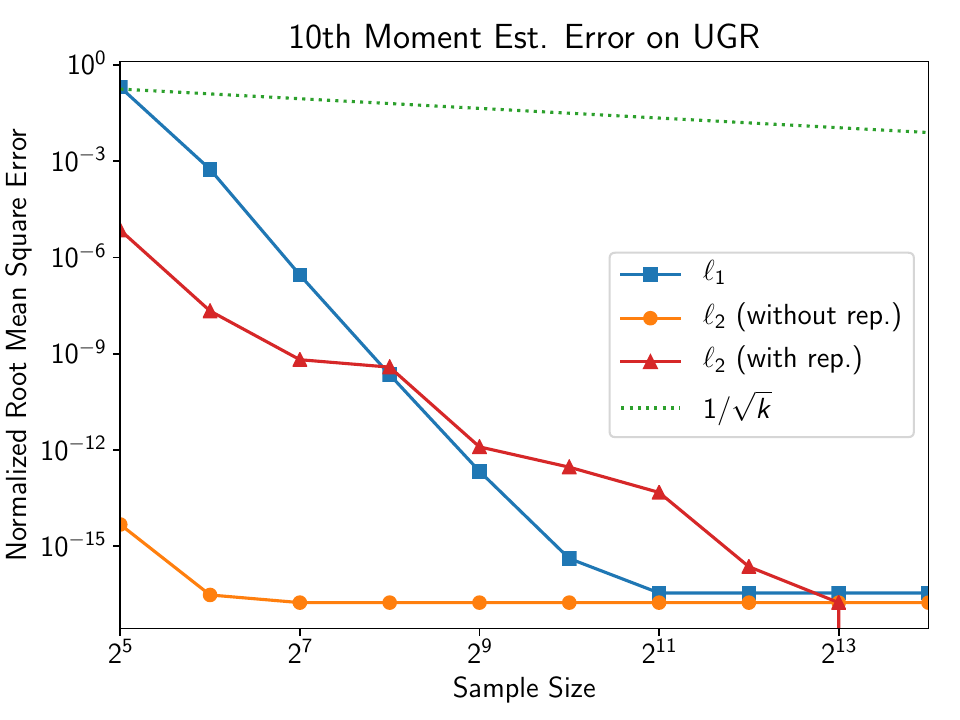}
\caption{Estimating 3rd and 10th moments on various datasets using PPSWOR ($\ell_1$ without replacement) and $\ell_2$ samples (with and without replacement).  The error is averaged over 50 repetition.}
\label{momentsest:fig}
\end{figure}

\subsection{Worst-Case Bound on Overhead}
The overhead \eqref{overheadpq:eq} of $\|\boldsymbol{w}/w_1\|^2_2/\|\boldsymbol{w}/w_1\|^p_p$ is the space factor increase needed for an
  $\ell_2$ sample to emulate an $\ell_p$ sample on the frequencies $\boldsymbol{w}$ (and accurately estimate the $p$-th moment).  A natural question is whether there is a better way to emulate an $\ell_p$ sampling with a polylogarithmic size sampling sketch. The following shows that in a sense an $\ell_2$ sample is the best we can do.
    \begin{lemma}
  Let $p \geq 2$. The overhead of emulating an $\ell_p$ sample using an $\ell_2$ sample on any frequency vector
  $\boldsymbol{w}$ of size $n$ is $n^{1-2/p}$.
  \end{lemma}
  \begin{proof} For any vector $v \in \mathbb{R}^n$, H\"older's inequality implies that $\|v\|_2 \leq n^{1/2-1/p}\|v\|_p$. As a result,
  \[\|v\|_2^2 \leq n^{1-2/p}\|v\|_p^2 \leq n^{1-2/p}\|v\|_p^p.\]
  Using Lemma~\ref{genratio:lemma}, we can bound the overhead of emulating an $\ell_p$ sample using an $\ell_2$ sample on any frequency vector
  $\boldsymbol{w}$ as
  \[\frac{\|\boldsymbol{w}/w_1\|^2_2}{\|\boldsymbol{w}/w_1\|^p_p} \leq n^{1-2/p}.\]
  \end{proof}
  This matches upper bounds on sketch size attained with dedicated sketches for $p$-th moment estimation~\cite{IW:stoc05,AndoniKO:FOCS11} and the worst-case lower bound of
  $\tilde{\Omega}(n^{1-2/p})$~\cite{AMS,li2013tight}. Interestingly, the worst case distributions that establish that bound are those where the most frequent key
is an $\ell_p$ heavy-hitter but not an $\ell_2$ heavy-hitter.

\subsection{Near-Uniform Frequency Distributions}
We showed that frequency distributions with heavy hitters are easy
for high moments and moreover, the validity of the result can be
certified.
Interestingly, the other extreme of near-uniform distributions (where
$w_1/w_n$ is bounded) is also easy.     Unfortunately, unlike the
case with heavy hitters, there is
no ``certificate'' to the validity of the emulation.
\begin{lemma}\label{lem:overhead-for-bounded-range}
Let $\boldsymbol{w}$ be a frequency distribution with support size
$n$. Then the overhead of using $\ell_1$ or $\ell_0$  sampling to
emulate $\ell_p$ sampling is at most $\left(\frac{w_1}{w_n}\right)^p$.
\end{lemma}
\begin{proof}
  We use Lemma~\ref{genratio:lemma} and lower bound the denominator
 $\left\| \frac{\boldsymbol{w}}{w_1}\right\|^p_p$. Note that for any $\boldsymbol{w}$ with support size $n$, $\left\| \frac{\boldsymbol{w}}{w_1}\right\|^1_1 \leq n$ and $\left\| \frac{\boldsymbol{w}}{w_1}\right\|_0 = n$. Now,
\[
\left\| \frac{\boldsymbol{w}}{w_1}\right\|^p_p = \sum_{i=1}^n{\left(\frac{w_i}{w_1}\right)^p} \geq \sum_{i=1}^n{\left(\frac{w_n}{w_1}\right)^p} = n \cdot \left(\frac{w_n}{w_1}\right)^p.
\]
The overhead for $q\in \{0,1\}$ is then
\[
\frac{\left\| \frac{\boldsymbol{w}}{w_1}\right\|^q_q}{\left\| \frac{\boldsymbol{w}}{w_1} \right\|^p_p } \leq \frac{n}{n \cdot \left(\frac{w_n}{w_1}\right)^p} = \left(\frac{w_1}{w_n}\right)^p.
\]
\end{proof}

\section{Universal Samples} \label{uni:sec}

In this section, we study combinations where the emulation is {\em universal}, that is,
the set of functions $F$  includes all monotonically non-decreasing functions of
frequencies. We denote the set of all monotonically non-decreasing functions by $M$.
Interestingly, there are sampling probabilities that provide universal
emulation for any frequency vector $\boldsymbol{w}$.
     \begin{lemma}\cite{MultiObjective}\label{MO:lemma}
Consider the probabilities $\boldsymbol{q}$ where the $i$-th most frequent key has
  $q_i=\frac{1}{iH_n}$ ($H_n = \sum_{i=1}^n\frac{1}{i}$ is the $n$-th harmonic number).  Then
a weighted sample by
  $\boldsymbol{q}$ is a universal emulator with overhead at most
$H_n$.
     \end{lemma}
     \begin{proof}
     Consider a monotonically  non-decreasing $f$ with respective PPS
     probabilities $\boldsymbol{p}$.  By definition, for the $i$-th
     most frequent key, $p_i = \frac{f(w_i)}{\|f(\boldsymbol{w}) \|_1} \leq \frac{f(w_i)}{f(w_1) + \ldots + f(w_i)}
     \leq \frac{1}{i}$ (note that for $j \leq i$, $f(w_i) \leq f(w_j)$).  Therefore, $p_i/q_i \leq \frac{1}{i} (i H_n)  = H_n$.
   \end{proof}
   This universal sampling, however, cannot be implemented with small
   (polylogarithmic size) sketches. This is because
$M$ includes functions that require large (polynomial size)
   sketches such as thresholds ($I_{w\geq T}$ for some fixed $T$)
   and high moments ($p>2$).  We therefore aim to find small sketches that
   provide universal emulation for a restricted $W$.

   For particular sampling probabilities  $\boldsymbol{q}$ and frequencies $\boldsymbol{w}$ we
consider the {\em universal emulation overhead} to be the overhead factor that will
allow the sample according to $\boldsymbol{q}$ to emulate weighted sampling with respect to
$f(\boldsymbol{w})$ for any $f\in M$ (the set of all monotonically  non-decreasing functions):
\begin{equation}
  \max_{f\in M} \max_i f(w_i)/(\|f(\boldsymbol{w})\|_1 q_i)
\end{equation}
Interestingly, the universal emulation overhead of $\boldsymbol{q}$ does not
depend on the particular $\boldsymbol{w}$.
\begin{lemma} \label{universaloverhead:lemma}
  The universal emulation overhead of $\boldsymbol{q}$ is at most
  \[\max_i   1/(i
    q_i).\]
    The bound $\max_i 1/(i q_i)$ is tight and always at least $H_n$ even when $W$ contains a single
    $\boldsymbol{w}$, as long as frequencies are distinct ($w_i>w_{i+1}$ for all $i$).
\end{lemma}
  \begin{proof}
Consider   $\boldsymbol{w}$.  We have
     \begin{eqnarray*}
\lefteqn{\max_{f\in M} \max_i f(w_i)/(\|f(\boldsymbol{w})\|_1 q_i) }\\
 &=&
  \max_i \max_{f\in M} f(w_i)/(\|f(\boldsymbol{w})\|_1 q_i)  \leq
  \max_i 1/(i q_i) \ .
     \end{eqnarray*}
As in Lemma~\ref{MO:lemma}, the last inequality follows
since $f$ is non-decreasing:  for all $j<i$ we must have $f(w_j) \geq f(w_i)$.
Therefore, $f(w_i)/\|f(\boldsymbol{w})\|_1 \leq 1/i$, and we get that the overhead is upper-bounded by $\max_i 1/(i q_i)$.

Note that for the threshold function $f(w) = I_{w \geq w_i}$, $f(w_i)/\|f(\boldsymbol{w})\|_1 = 1/i$ (if $w_{i+1} < w_i$). Hence, when the frequencies are distinct, the bound $\max_i 1/(i q_i)$ on the overhead is tight.

To conclude the proof, we show that $\max_i 1/(i q_i) \geq H_n$. Assume, for the sake of contradiction, that $\frac{1}{i q_i} < H_n$ for all $i$. Then, $\frac{1}{i} < q_i H_n$, and by summing over all $i$, we get $\sum_{i=1}^n {\frac{1}{i}} < \sum_{i=1}^n {q_i H_n} = H_n$, which is a contradiction. Therefore, we get that $\max_i 1/(i q_i) \geq H_n$.
\end{proof}
\begin{remark}
Lemmas \ref{MO:lemma} and \ref{universaloverhead:lemma} are connected in the following way. Suppose we are given sampling probabilities $\boldsymbol{q}$ and would like to compute their universal emulation overhead. From Lemma~\ref{universaloverhead:lemma}, we know that the overhead is bounded by $\max_i 1/(i q_i)$. Alternatively, Lemma~\ref{MO:lemma} tells us that the sampling probabilities $\{\frac{1}{iH_n}\}$ have universal emulation overhead $H_n$. The overhead of emulating the sampling probabilities $\{\frac{1}{iH_n}\}$ using $\boldsymbol{q}$ is $\max_i \frac{1}{iH_n q_i}$. Since the overhead accumulates multiplicatively (Remark~\ref{multaccum:rem}), we can combine the last two statements, and conclude the universal emulation overhead of $\boldsymbol{q}$ is upper-bounded by $\max_i \frac{1}{iH_n q_i} \cdot H_n = \max_i 1/(i q_i)$, as in Lemma~\ref{universaloverhead:lemma}.
\end{remark}

We can similarly consider, for given sampling probabilities $\boldsymbol{q}$,
the
{\em universal estimation overhead}, which is the overhead needed for estimating full
$f$-statistics for all $f \in M$.
As discussed in
Section~\ref{sec:preliminaries}, estimation of the full $f$-statistics is a weaker
requirement than emulation. Hence, for any particular $\boldsymbol{q}$ the estimation overhead can
be lower than the emulation overhead.  The estimation overhead,
however, is still at least $H_n$.
 \begin{lemma} \label{universalestoverhead:lemma}
    The universal estimation overhead for estimating all monotone $f$-statistics
    for
    $\boldsymbol{q}$ is
    \[\max_i   \frac{1}{i^2} \sum_{j=1}^i \frac{1}{q_j}\ .\]
\end{lemma}
\begin{proof}
  By Corollary~\ref{cor:estimation-overhead}, the universal estimation overhead with frequencies $\boldsymbol{w}$ is
     \begin{equation}\label{eq:uni-est-from-def}
\max_{f\in M} \sum_j \frac{(f(w_j))^2}{\|f(\boldsymbol{w})\|_1^2}
\frac{1}{q_j} \enspace .
     \end{equation}
     It suffices to consider $f$ that are threshold functions (follows from the inequality $\frac{a_1 + \ldots + a_k}{b_1 + \ldots + b_k} \leq \max_{j=1,\ldots,k}\frac{a_j}{b_j}$ as each $f \in M$ can be represented as a combination of threshold functions). Specifically, consider the threshold function $f(w)=I_{w \geq w_i}$ for some $i$. As in Lemmas \ref{MO:lemma} and \ref{universaloverhead:lemma}, the
     expression \eqref{eq:uni-est-from-def} for the threshold function $f(w)=I_{w \geq w_i}$ has
     $f(w_j)/\|f(\boldsymbol{w})\|_1 = 1/i$ for $j\leq i$ and $0$
     otherwise (assuming all the frequencies in $\boldsymbol{w}$ are distinct).
     We get that the sum is $\frac{1}{i^2} \sum_{j=1}^i
     \frac{1}{q_j}$.  The claim follows from taking the maximum over
     all threshold functions.
\end{proof}

\subsection{Universal Emulation Using Off-the-Shelf Sampling Sketches}
In our context, the probabilities $\boldsymbol{q}$ are not
something we directly control but rather emerge as an artifact of applying a
certain sampling scheme to a dataset with certain frequencies
$\boldsymbol{w}$.  We will explore the universal overhead of
the $\boldsymbol{q}$  obtained
when applying off-the-shelf schemes (see
Section~\ref{offtheshelf:sec})  to $\Zipf$ frequencies and to our datasets.

\begin{lemma}\label{lem:univ-zipf}
For a frequency vector $\boldsymbol{w}$ that is $\Zipf[\alpha,n]$, $\ell_p$ sampling with
     $p=1/\alpha$ is a universal emulator with (optimal) overhead $H_n$.
\end{lemma}
\begin{proof}
Consider the sampling probability $q_i$ for the $i$-th key (with frequency $w_i$) when using $\ell_{1/\alpha}$ sampling. From Definition~\ref{def:zipf}, we know that $w_i = \frac{w_1}{i^\alpha}$. We get that $w_i^{1/\alpha} = \frac{w_1^{1/\alpha}}{i}$, and consequently,
\[
q_i = \frac{w_i^{1/\alpha}}{\sum_{j=1}^n{w_j^{1/\alpha}}} = \frac{w_1^{1/\alpha} / i}{\sum_{j=1}^n{w_1^{1/\alpha}/j}} = \frac{1 / i}{\sum_{j=1}^n{1/j}} = \frac{1}{i H_n}.
\]
The sampling probabilities $\boldsymbol{q}$ are the same as in Lemma~\ref{MO:lemma}. As a result, we get that applying $\ell_{1/\alpha}$ sampling to the frequency vector $\boldsymbol{w}$ has universal emulation overhead of $H_n$.
\end{proof}
Interestingly, for $\alpha \geq 1/2$, universal emulation as in Lemma~\ref{lem:univ-zipf} is attained by $\ell_p$ sampling with
$p\leq 2$, which can be implemented with polylogarithmic-size sketches.
Note that we match here a different $\ell_p$ sample for each possible $\Zipf$ parameter $\alpha$ of the
data frequencies.   A sampling scheme that
emulates $\ell_p$ sampling for a range $[p_1,p_2]$ of
$p$ values with some overhead $h$ will be a
universal emulator with overhead $h H_n$ for $\Zipf[\alpha]$ for
$\alpha\in [1/p_2,1/p_1]$ (see Remark~\ref{multaccum:rem}).
One such sampling scheme with polylogarithmic-size sketches was provided in
\cite{CapSampling, MultiObjective}.  The sample emulates all
concave sublinear functions, including capping functions
$f(w)=\min\{w,T\}$ for $T>0$ and low moments with $p\in [0,1]$, with
$O(\log n)$ overhead.

We next express a condition on frequency distributions under which a
multi-objective concave-sublinear sample provides
a universal emulator.
The condition is that for all $i$, the weight of the $i$-th most frequent key is at least
$c/i$ times the weight of the tail from $i$.
\begin{lemma}
    Let $\boldsymbol{w}$ be such that
      $\min_i \frac{i w_i}{\sum_{j=i+1}^n w_j} \geq c$.  Then a sample
      that emulates all concave-sublinear functions
with overhead $h'$ is a universal emulator for $\boldsymbol{w}$
with overhead at most $h'(1+1/c)$.
\end{lemma}
\begin{proof}
Denote the sampling probabilities of the given sampling method (that emulates all concave-sublinear functions) by $\boldsymbol{q}$. By Lemma~\ref{universaloverhead:lemma}, we need to bound $\max_i 1/(i q_i)$. First, consider the $i$-th largest frequency $w_i$, and assume that we are sampling according to the capping function $f(w) = \min\{w, w_i\}$ instead of $\boldsymbol{q}$. Later we will remove this assumption (and pay additional overhead $h'$) since the given sampling scheme emulates all capping functions. When sampling according to $f(w) = \min\{w, w_i\}$, the sampling probability of the $i$-th key is $\frac{f(w_i)}{\sum_{j=1}^n f(w_j)}$. Our goal is to bound $\frac{1}{i\frac{f(w_i)}{\sum_{j=1}^n f(w_j)}}$. Using the condition in the statement of the lemma,
      \[
      \frac{f(w_i)}{\sum_{j=1}^n f(w_j)} = \frac{w_i}{i w_i + \sum_{j=i+1}^n w_j}
     \geq  \frac{w_i}{i w_i + \frac{i}{c} w_i} = \frac{1}{i(1+1/c)}\enspace .
      \]
Hence,
\[
\frac{1}{i\frac{f(w_i)}{\sum_{j=1}^n f(w_j)}} \leq \frac{1}{i\frac{1}{i(1+1/c)}} = 1 + \frac{1}{c}\enspace .
\]

We now remove the assumption that we sample according to $f(w) = \min\{w, w_i\}$. Note that the given sampling scheme emulates all concave sublinear functions (a family which includes all capping functions) with overhead $h'$. From the definition of emulation overhead, for $f(w) = \min\{w, w_i\}$ we get that
\[
\frac{f(w_i)}{q_i\sum_{j=1}^n f(w_j)} \leq h'\enspace .
\]
Now,
\[
\frac{1}{i q_i} = \frac{f(w_i)}{q_i\sum_{j=1}^n f(w_j)} \cdot \frac{1}{i\frac{f(w_i)}{\sum_{j=1}^n f(w_j)}} \leq h'\left(1 + \frac{1}{c}\right)\enspace .
\]
Since the bound applies to all $i$, we get that the universal emulation overhead is $\max_i 1/(i q_i) \leq h'(1+1/c)$, as desired.
\end{proof}
       Interestingly, for high moments to be ``easy'' it sufficed to
       have a heavy hitter. For universal emulation we need to bound from below the relative weight of each key with respect to the remaining tail.

\paragraph{Experimental results.}
Table~\ref{table:overhead} reports the universal emulation overhead on our datasets with
$\ell_1$, $\ell_2$, and multi-objective concave-sublinear sampling probabilities.
We observe that while $\ell_2$ sampling emulates high moments
extremely well, its universal overhead is very large due to poor
emulation of ``slow growth'' functions.  The better universal overhead $h$ of $\ell_1$ and concave-sublinear samples satisfies $h \in [92,1673]$. It is practically meaningful as it is in the regime where $h\epsilon^{-2} \ll n$ (the number of keys is significantly larger than the sample size needed to get normalized root mean squared error $\epsilon$).

The universal emulation overhead was computed using
Lemma~\ref{universaloverhead:lemma} with respect to base PPS
probabilities for the off-the-shelf sampling schemes (given in Section~\ref{offtheshelf:sec}).

\section{Conclusion}
In this work, we studied composable sampling sketches under two beyond-worst-case scenarios. In the first, we assumed additional information about the input distribution in the form of advice. We designed and analyzed a sampling sketch based on the advice, and demonstrated its performance on real-world datasets.

In the second scenario, we proposed a framework where the performance and statistical guarantees of sampling sketches were analyzed in terms of supported {\em frequency-function combinations}.
We demonstrated analytically and empirically that sketches originally
designed to sample according to ``easy'' functions of frequencies on ``hard'' frequency
distributions turned out to be accurate for sampling according to ``hard'' functions of frequencies
on ``practical'' frequency distributions.
In particular, on ``practical''  distributions we could accurately
approximate high frequency moments ($p>2$) and the rank versus
frequency distribution using small composable sketches.

\section*{Acknowledgments}
We are grateful to the authors of \cite{hsuIKV:ICLR2019}, especially Chen-Yu Hsu and Ali Vakilian, for sharing their data, code, and predictions with us. We thank Ravi Kumar and Robert Krauthgamer for helpful discussions.
The work of Ofir Geri was partially supported by Moses Charikar's Simons Investigator Award.
The work of Rasmus Pagh was supported by Investigator Grant 16582, Basic Algorithms Research Copenhagen (BARC), from the VILLUM Foundation.

\bibliographystyle{plain}
\bibliography{new_thesis_ref}

\begin{thebibliography}{10}

\bibitem{AamandIV:arxiv2019}
Anders Aamand, Piotr Indyk, and Ali Vakilian.
\newblock ({L}earned) frequency estimation algorithms under {Z}ipfian
  distribution.
\newblock {\em arXiv}, abs/1908.05198, 2019.
\newblock \url{http://arxiv.org/abs/1908.05198}.

\bibitem{AgarwalMergeable}
Pankaj~K. Agarwal, Graham Cormode, Zengfeng Huang, Jeff~M. Phillips, Zhewei
  Wei, and Ke~Yi.
\newblock Mergeable summaries.
\newblock {\em {ACM} Transactions on Database Systems}, 38(4), December 2013.

\bibitem{alon2002tracking}
Noga Alon, Phillip~B Gibbons, Yossi Matias, and Mario Szegedy.
\newblock Tracking join and self-join sizes in limited storage.
\newblock {\em Journal of Computer and System Sciences}, 64(3):719--747, 2002.

\bibitem{AMS}
Noga Alon, Yossi Matias, and Mario Szegedy.
\newblock The space complexity of approximating the frequency moments.
\newblock {\em Journal of Computer and System Sciences}, 58(1):137 -- 147,
  1999.

\bibitem{AndoniKO:FOCS11}
Alexandr Andoni, Robert Krauthgamer, and Krzysztof Onak.
\newblock Streaming algorithms via precision sampling.
\newblock In {\em {IEEE} 52nd Annual Symposium on Foundations of Computer
  Science, {FOCS} 2011}. {IEEE} Computer Society, 2011.

\bibitem{BravermanOstro:STOC2010}
Vladimir Braverman and Rafail Ostrovsky.
\newblock Zero-one frequency laws.
\newblock In {\em Proceedings of the 42nd {ACM} Symposium on Theory of
  Computing, {STOC} 2010}, pages 281--290. {ACM}, 2010.

\bibitem{CAIDA}
{CAIDA}.
\newblock The {CAIDA} {UCSD} anonymized internet traces 2016 --- 2016/01/21
  13:29:00 {UTC}.
\newblock https://www.caida.org/data/passive/passive\_dataset.xml, 2016.

\bibitem{Cha82}
M.~T. Chao.
\newblock A general purpose unequal probability sampling plan.
\newblock {\em Biometrika}, 69(3):653--656, 1982.

\bibitem{ccf:icalp2002}
Moses Charikar, Kevin~C. Chen, and Martin Farach{-}Colton.
\newblock Finding frequent items in data streams.
\newblock In {\em Automata, Languages and Programming, 29th International
  Colloquium, {ICALP} 2002}, volume 2380 of {\em Lecture Notes in Computer
  Science}, pages 693--703. Springer, 2002.

\bibitem{MultiObjective}
Edith Cohen.
\newblock Multi-objective weighted sampling.
\newblock In {\em 2015 Third IEEE Workshop on Hot Topics in Web Systems and
  Technologies (HotWeb)}, pages 13--18, 2015.

\bibitem{CapSampling}
Edith Cohen.
\newblock Stream sampling framework and application for frequency cap
  statistics.
\newblock {\em ACM Transactions on Algorithms}, 14(4):52:1--52:40, 2018.

\bibitem{CCD:sigmetrics12}
Edith Cohen, Graham Cormode, and Nick~G. Duffield.
\newblock Don't let the negatives bring you down: sampling from streams of
  signed updates.
\newblock In {\em {ACM} {SIGMETRICS/PERFORMANCE} Joint International Conference
  on Measurement and Modeling of Computer Systems, {SIGMETRICS} '12}, pages
  343--354. {ACM}, 2012.

\bibitem{varopt_full:CDKLT10}
Edith Cohen, Nick Duffield, Haim Kaplan, Carsten Lund, and Mikkel Thorup.
\newblock Efficient stream sampling for variance-optimal estimation of subset
  sums.
\newblock {\em SIAM Journal on Computing}, 40(5):1402--1431, 2011.

\bibitem{CohenGeri:NeurIPS2019}
Edith Cohen and Ofir Geri.
\newblock Sampling sketches for concave sublinear functions of frequencies.
\newblock In {\em Advances in Neural Information Processing Systems, {NeurIPS}
  2019}, volume~32. Curran Associates, Inc., 2019.

\bibitem{bottomk07:ds}
Edith Cohen and Haim Kaplan.
\newblock Summarizing data using bottom-k sketches.
\newblock In {\em Proceedings of the Twenty-Sixth Annual {ACM} Symposium on
  Principles of Distributed Computing, {PODC} 2007}, pages 225--234. {ACM},
  2007.

\bibitem{bottomk:VLDB2008}
Edith Cohen and Haim Kaplan.
\newblock Tighter estimation using bottom k sketches.
\newblock volume~1, pages 213--224, 2008.

\bibitem{multiw:VLDB2009}
Edith Cohen, Haim Kaplan, and Subhabrata Sen.
\newblock Coordinated weighted sampling for estimating aggregates over multiple
  weight assignments.
\newblock {\em Proceedings of the VLDB Endowment}, 2(1):646--657, 2009.

\bibitem{CohenPW:NeurIPS2020}
Edith Cohen, Rasmus Pagh, and David Woodruff.
\newblock Wor and $p$'s: Sketches for $\ell_p$-sampling without replacement.
\newblock In {\em Advances in Neural Information Processing Systems, {NeurIPS}
  2020}, volume~33, pages 21092--21104. Curran Associates, Inc., 2020.

\bibitem{CormodeMuthu:2005}
Graham Cormode and S.~Muthukrishnan.
\newblock An improved data stream summary: the count-min sketch and its
  applications.
\newblock {\em Journal of Algorithms}, 55(1):58--75, 2005.

\bibitem{DLT:jacm07}
Nick~G. Duffield, Carsten Lund, and Mikkel Thorup.
\newblock Priority sampling for estimating arbitrary subset sums.
\newblock {\em Journal of the {ACM}}, 54(6):32–es, 2007.

\bibitem{EdenRS:siamdm2019}
Talya Eden, Dana Ron, and C.~Seshadhri.
\newblock Sublinear time estimation of degree distribution moments: The
  arboricity connection.
\newblock {\em {SIAM} Journal on Discrete Mathematics}, 33(4):2267--2285, 2019.

\bibitem{EV:ATAP02}
Cristian Estan and George Varghese.
\newblock New directions in traffic measurement and accounting: Focusing on the
  elephants, ignoring the mice.
\newblock {\em {ACM} Transactions on Computer Systems}, 21(3):270--313, 2003.

\bibitem{hyperloglog:2007}
Philippe Flajolet, {\'E}ric Fusy, Olivier Gandouet, and Fr{\'e}d{\'e}ric
  Meunier.
\newblock {HyperLogLog}: The analysis of a near-optimal cardinality estimation
  algorithm.
\newblock In {\em Analysis of Algorithms (AofA)}, pages 137--156. DMTCS, 2007.

\bibitem{FlajoletMartin85}
Philippe Flajolet and G.~Nigel Martin.
\newblock Probabilistic counting algorithms for data base applications.
\newblock {\em Journal of Computer and System Sciences}, 31(2):182--209, 1985.

\bibitem{FriezeKV:JACM2004}
Alan Frieze, Ravi Kannan, and Santosh Vempala.
\newblock Fast monte-carlo algorithms for finding low-rank approximations.
\newblock {\em Journal of the ACM}, 51(6), 2004.

\bibitem{GM:sigmod98}
Phillip~B. Gibbons and Yossi Matias.
\newblock New sampling-based summary statistics for improving approximate query
  answers.
\newblock In {\em Proceedings {ACM} {SIGMOD} International Conference on
  Management of Data, {SIGMOD} 1998}, pages 331--342. {ACM}, 1998.

\bibitem{HT52}
D.~G. Horvitz and D.~J. Thompson.
\newblock A generalization of sampling without replacement from a finite
  universe.
\newblock {\em Journal of the American Statistical Association},
  47(260):663--685, 1952.

\bibitem{hsuIKV:ICLR2019}
Chen-Yu Hsu, Piotr Indyk, Dina Katabi, and Ali Vakilian.
\newblock Learning-based frequency estimation algorithms.
\newblock In {\em International Conference on Learning Representations,
  {ICLR}}, 2019.

\bibitem{indyk:stable}
Piotr Indyk.
\newblock Stable distributions, pseudorandom generators, embeddings, and data
  stream computation.
\newblock {\em Journal of the {ACM}}, 53(3):307--323, 2006.

\bibitem{IndykVY:NeurIPS2019}
Piotr Indyk, Ali Vakilian, and Yang Yuan.
\newblock Learning-based low-rank approximations.
\newblock In {\em Advances in Neural Information Processing Systems, {NeurIPS}
  2019}, volume~32. Curran Associates, Inc., 2019.

\bibitem{IW:stoc05}
Piotr Indyk and David~P. Woodruff.
\newblock Optimal approximations of the frequency moments of data streams.
\newblock In {\em Proceedings of the 37th Annual {ACM} Symposium on Theory of
  Computing, {STOC} 2005}, pages 202--208. {ACM}, 2005.

\bibitem{JayaramW:Focs2018}
Rajesh Jayaram and David~P. Woodruff.
\newblock Perfect lp sampling in a data stream.
\newblock In {\em 59th {IEEE} Annual Symposium on Foundations of Computer
  Science, {FOCS} 2018}, pages 544--555. {IEEE} Computer Society, 2018.

\bibitem{jiangLLRW:ICLR2020}
Tanqiu Jiang, Yi~Li, Honghao Lin, Yisong Ruan, and David~P. Woodruff.
\newblock Learning-augmented data stream algorithms.
\newblock In {\em International Conference on Learning Representations,
  {ICLR}}, 2020.

\bibitem{KraskaBCDP:sigmod2018}
Tim Kraska, Alex Beutel, Ed~H. Chi, Jeffrey Dean, and Neoklis Polyzotis.
\newblock The case for learned index structures.
\newblock In {\em Proceedings of the 2018 International Conference on
  Management of Data, {SIGMOD} Conference 2018}, pages 489--504. {ACM}, 2018.

\bibitem{li2013tight}
Yi~Li and David~P. Woodruff.
\newblock A tight lower bound for high frequency moment estimation with small
  error.
\newblock In {\em Approximation, Randomization, and Combinatorial Optimization.
  Algorithms and Techniques - 16th International Workshop, {APPROX} 2013, and
  17th International Workshop, {RANDOM} 2013}, volume 8096 of {\em Lecture
  Notes in Computer Science}, pages 623--638. Springer, 2013.

\bibitem{LiuBEKBFS:Sigcomm2019}
Zaoxing Liu, Ran Ben{-}Basat, Gil Einziger, Yaron Kassner, Vladimir Braverman,
  Roy Friedman, and Vyas Sekar.
\newblock Nitrosketch: robust and general sketch-based monitoring in software
  switches.
\newblock In {\em Proceedings of the {ACM} Special Interest Group on Data
  Communication, {SIGCOMM} 2019}, pages 334--350, 2019.

\bibitem{LiuMVSB:sigcomm2016}
Zaoxing Liu, Antonis Manousis, Gregory Vorsanger, Vyas Sekar, and Vladimir
  Braverman.
\newblock One sketch to rule them all: Rethinking network flow monitoring with
  univmon.
\newblock In {\em Proceedings of the {ACM} {SIGCOMM} 2016 Conference}, pages
  101--114. {ACM}, 2016.

\bibitem{UGR}
Gabriel Maci\'{a}-Fern\'{a}ndez, Jos\'{e} Camacho, Roberto
  Mag\'{a}n-Carri\'{o}n, Pedro Garc\'{i}a-Teodoro, and Roberto Ther\'{o}n.
\newblock {UGR}'16: A new dataset for the evaluation of cyclostationarity-based
  network idss.
\newblock {\em Computers \& Security}, 73:411 -- 424, 2018.

\bibitem{MM:vldb2002}
Gurmeet~Singh Manku and Rajeev Motwani.
\newblock Approximate frequency counts over data streams.
\newblock In {\em Proceedings of the 28th International Conference on Very
  Large Data Bases}, VLDB '02, page 346–357. VLDB Endowment, 2002.

\bibitem{mcgregor2016better}
Andrew McGregor, Sofya Vorotnikova, and Hoa~T. Vu.
\newblock Better algorithms for counting triangles in data streams.
\newblock In {\em Proceedings of the 35th {ACM} {SIGMOD-SIGACT-SIGAI} Symposium
  on Principles of Database Systems, {PODS} 2016}, pages 401--411. {ACM}, 2016.

\bibitem{pmlr-v54-mcmahan17a}
Brendan McMahan, Eider Moore, Daniel Ramage, Seth Hampson, and Blaise~Aguera
  y~Arcas.
\newblock {Communication-Efficient Learning of Deep Networks from Decentralized
  Data}.
\newblock In {\em Proceedings of the 20th International Conference on
  Artificial Intelligence and Statistics}, volume~54 of {\em Proceedings of
  Machine Learning Research}, pages 1273--1282. PMLR, 2017.

\bibitem{spacesaving:ICDT2005}
Ahmed Metwally, Divyakant Agrawal, and Amr El~Abbadi.
\newblock Efficient computation of frequent and top-k elements in data streams.
\newblock In {\em Proceedings of the 10th International Conference on Database
  Theory, ICDT'05}, page 398–412. Springer-Verlag, 2005.

\bibitem{MisraGries:1982}
J.~Misra and David Gries.
\newblock Finding repeated elements.
\newblock {\em Science of Computer Programming}, 2(2):143--152, 1982.

\bibitem{MonemizadehWoodruff}
Morteza Monemizadeh and David~P. Woodruff.
\newblock 1-pass relative-error lp-sampling with applications.
\newblock In {\em Proceedings of the Twenty-First Annual ACM-SIAM Symposium on
  Discrete Algorithms}, SODA '10, page 1143–1160, USA, 2010. Society for
  Industrial and Applied Mathematics.

\bibitem{Ohlsson_SPS:1998}
Esbj\"{o}rn Ohlsson.
\newblock Sequential poisson sampling.
\newblock {\em Journal of Official Statistics}, 14(2):149--162, 1998.

\bibitem{StackOverflow}
Ashwin Paranjape, Austin~R. Benson, and Jure Leskovec.
\newblock Motifs in temporal networks.
\newblock In {\em Proceedings of the Tenth ACM International Conference on Web
  Search and Data Mining}, WSDM ’17, pages 601–--610. Association for
  Computing Machinery, 2017.

\bibitem{AOL}
Greg Pass, Abdur Chowdhury, and Cayley Torgeson.
\newblock A picture of search.
\newblock In {\em Proceedings of the 1st International Conference on Scalable
  Information Systems}, InfoScale ’06, pages 1--es. Association for Computing
  Machinery, 2006.

\bibitem{Rosen1972:successive}
Bengt Ros{\'e}n.
\newblock Asymptotic theory for successive sampling with varying probabilities
  without replacement, i.
\newblock {\em The Annals of Mathematical Statistics}, 43(2):373--397, 1972.

\bibitem{Rosen1997a}
Bengt Ros{\'e}n.
\newblock Asymptotic theory for order sampling.
\newblock {\em Journal of Statistical Planning and Inference}, 62(2):135--158,
  1997.

\bibitem{DLTVarAnalysis}
Mario Szegedy.
\newblock The {DLT} priority sampling is essentially optimal.
\newblock In {\em Proceedings of the Thirty-Eighth Annual ACM Symposium on
  Theory of Computing}, STOC '06, page 150–158. {ACM}, 2006.

\end{thebibliography}

\appendix
\section{Additional Experiments} \label{actualmore:sec}
\begin{figure*}[t]
\centering
  \includegraphics[width=0.45\textwidth]{aol_day50_ell_3.pdf}
  \includegraphics[width=0.45\textwidth]{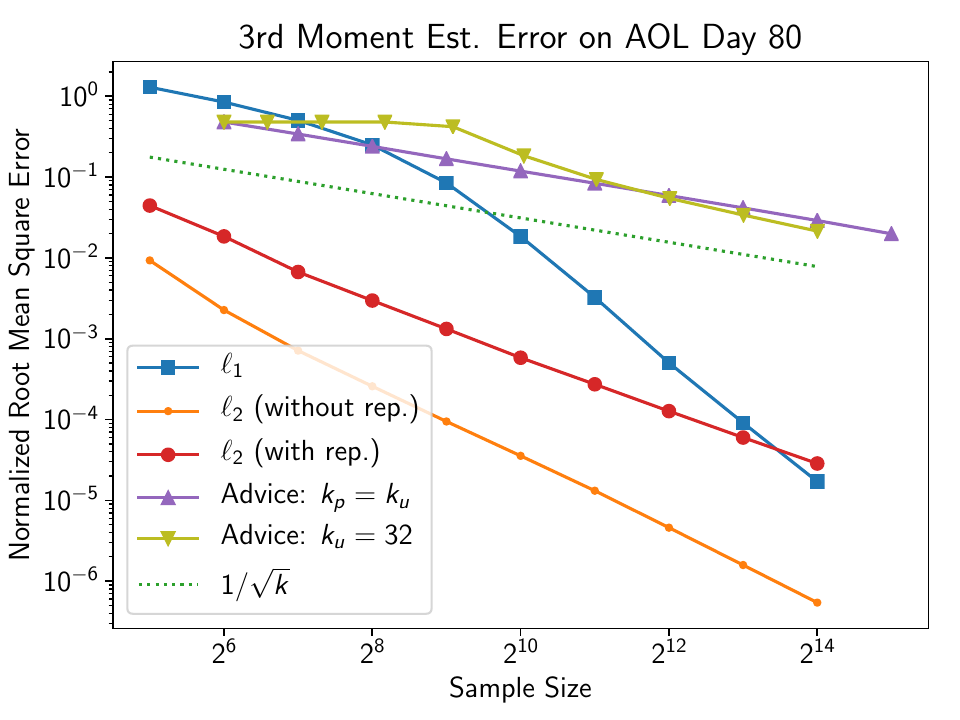}
  \includegraphics[width=0.45\textwidth]{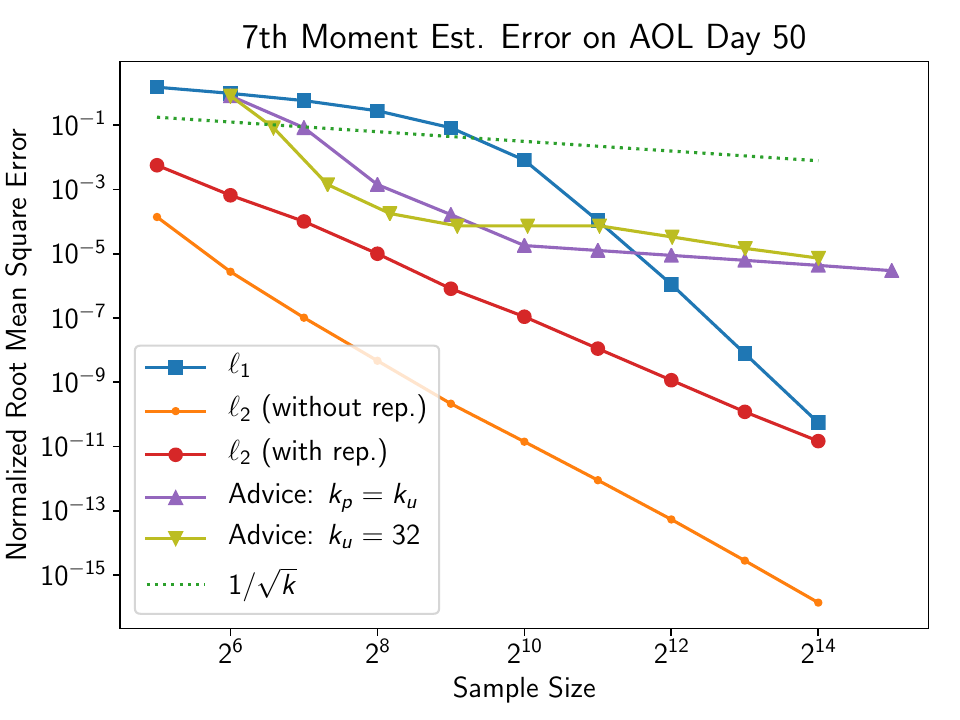}
  \includegraphics[width=0.45\textwidth]{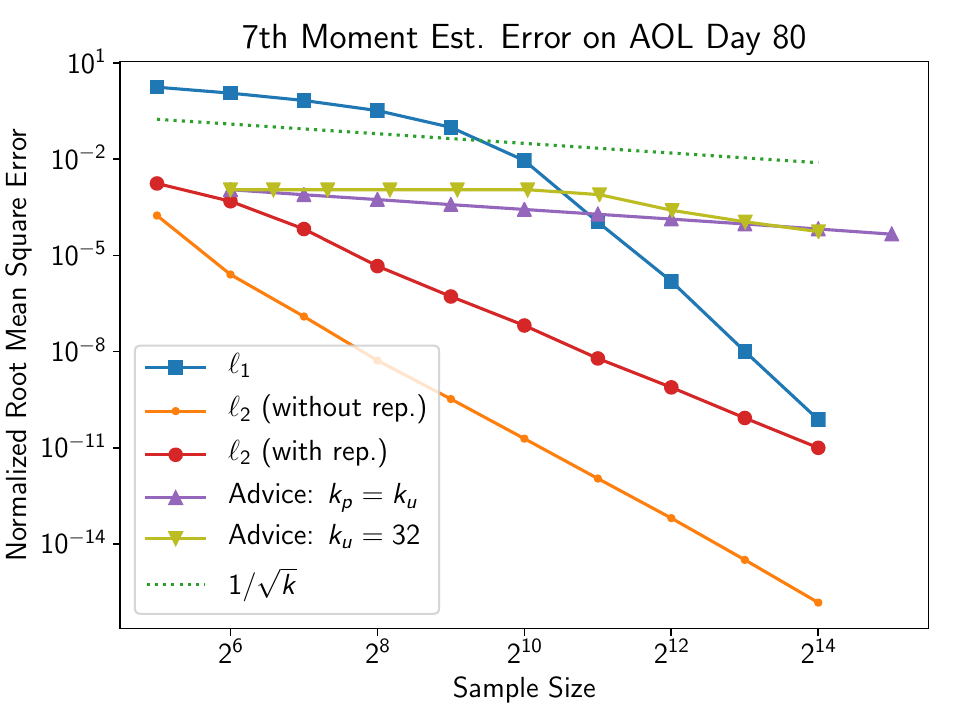}
  \includegraphics[width=0.45\textwidth]{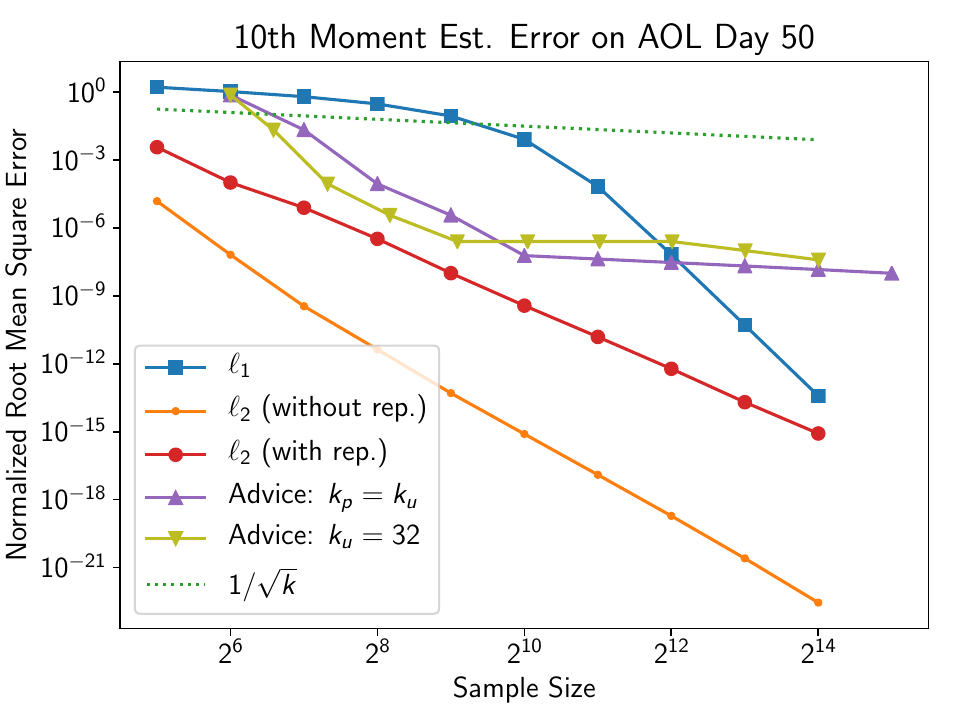}
  \includegraphics[width=0.45\textwidth]{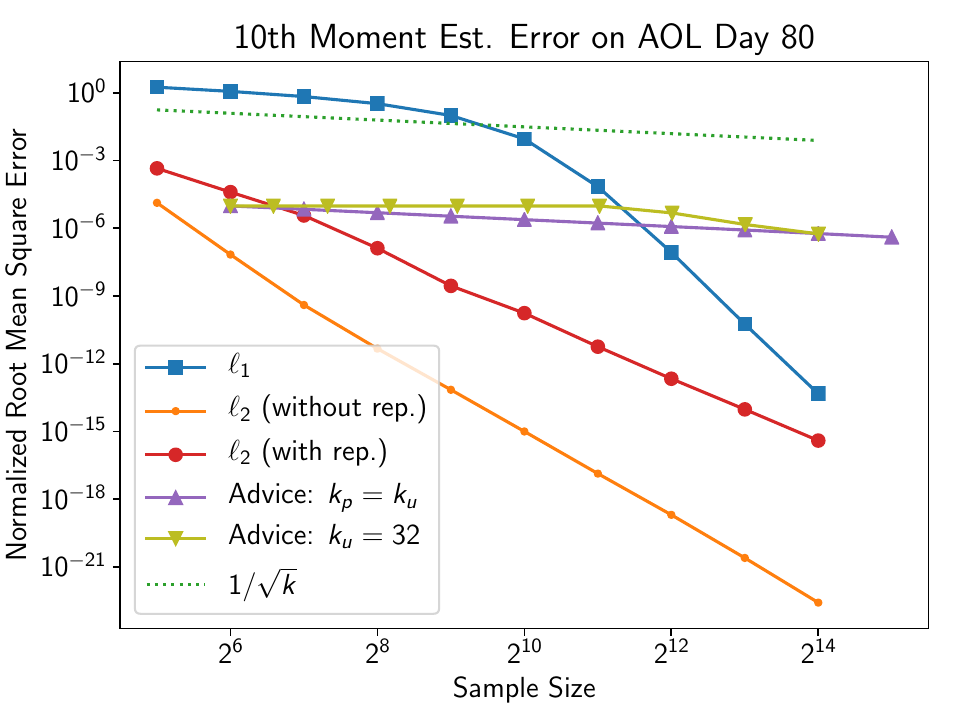}
\caption{NRMSE for estimating the 3rd, 7th, and 10th moments on the AOL dataset (days 50 and 80 with learned
  advice from~\cite{hsuIKV:ICLR2019}).}
\label{advicemore:fig}
\end{figure*}
\begin{figure*}[t]
\centering
  \includegraphics[width=0.45\textwidth]{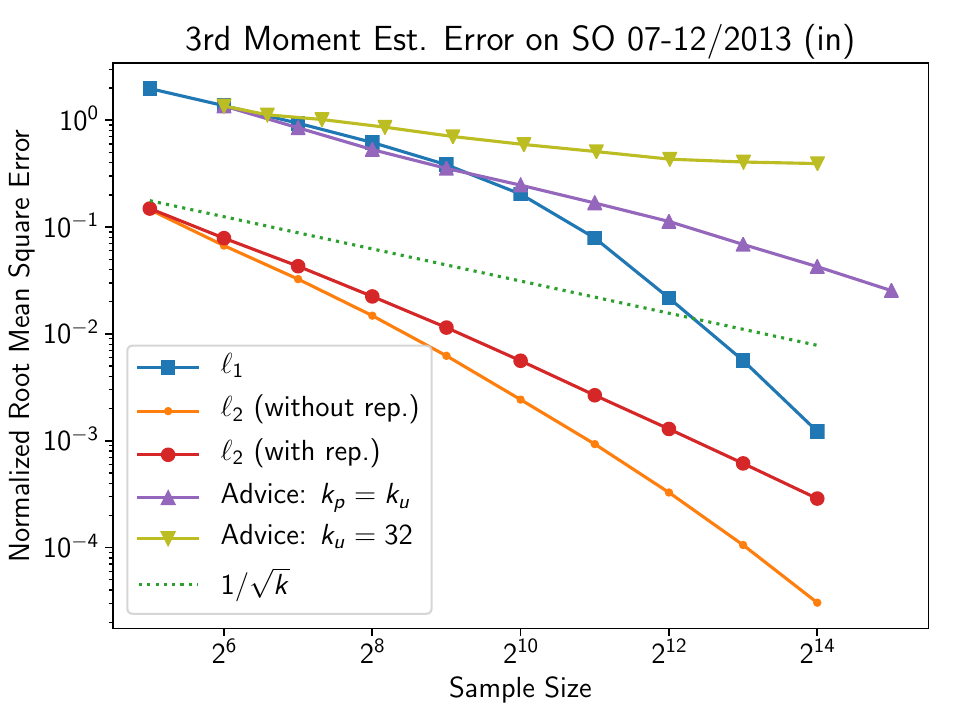}
  \includegraphics[width=0.45\textwidth]{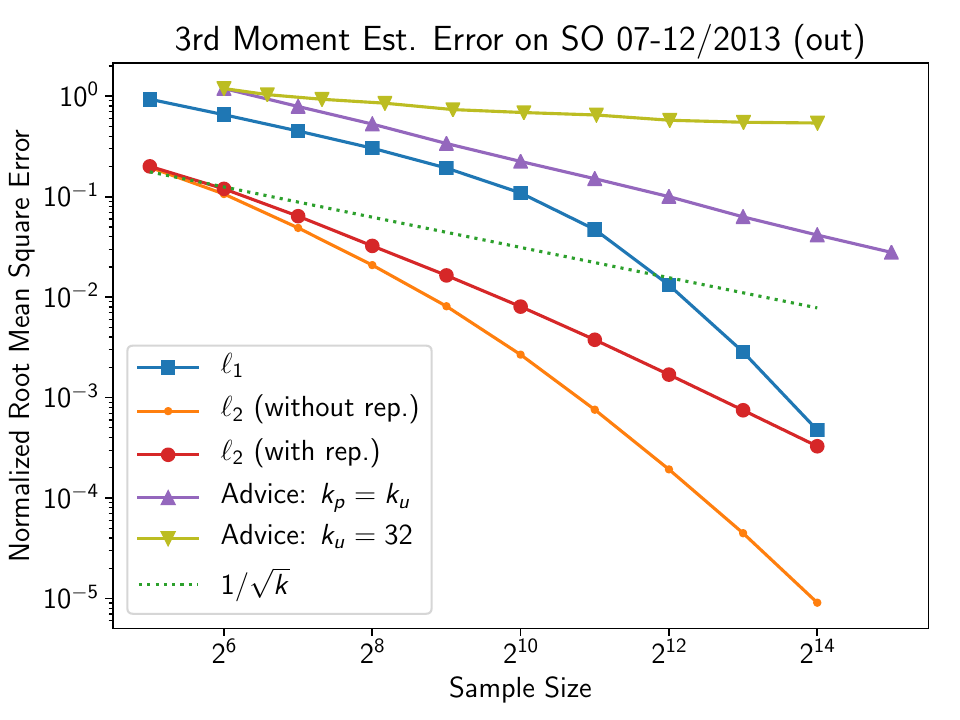}
  \includegraphics[width=0.45\textwidth]{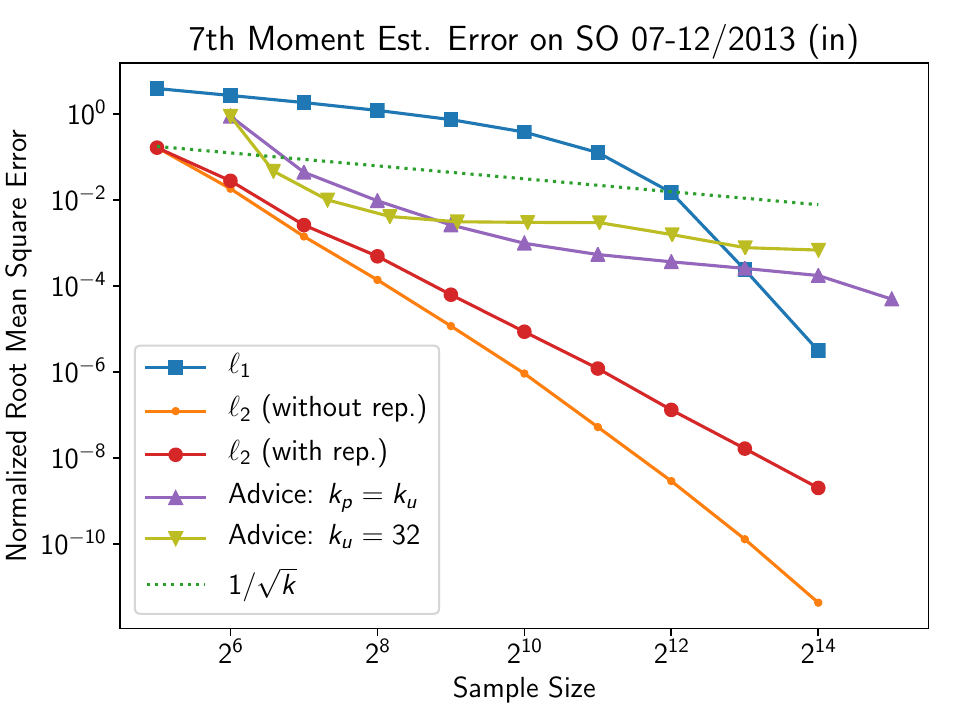}
  \includegraphics[width=0.45\textwidth]{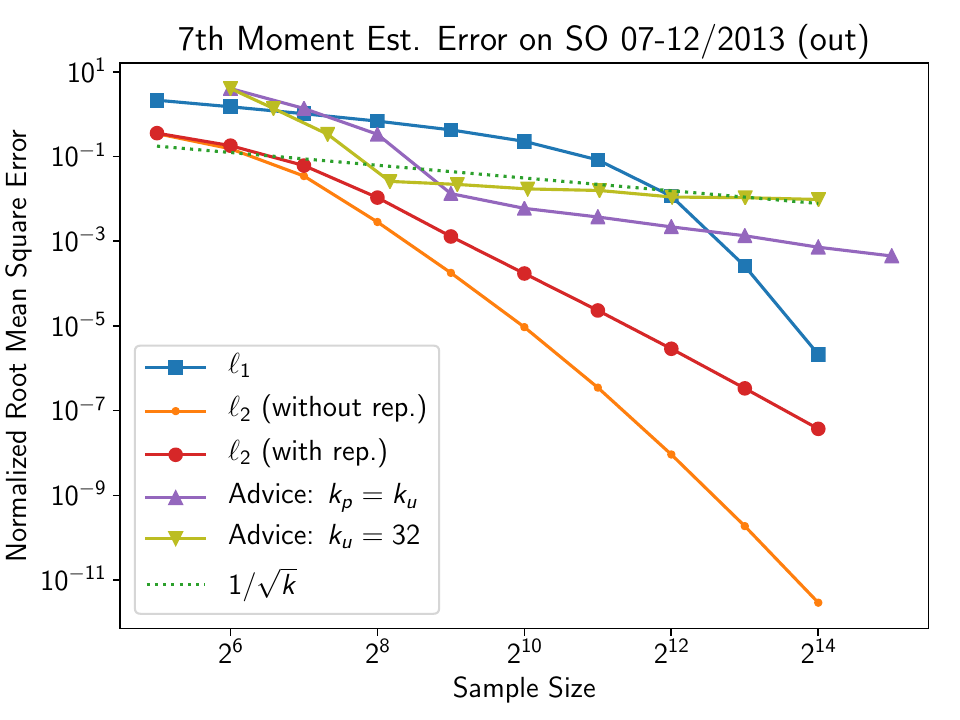}
  \includegraphics[width=0.45\textwidth]{stack_overflow_2013_07-12_correct_edge_dir_in_ell_10.pdf}
  \includegraphics[width=0.45\textwidth]{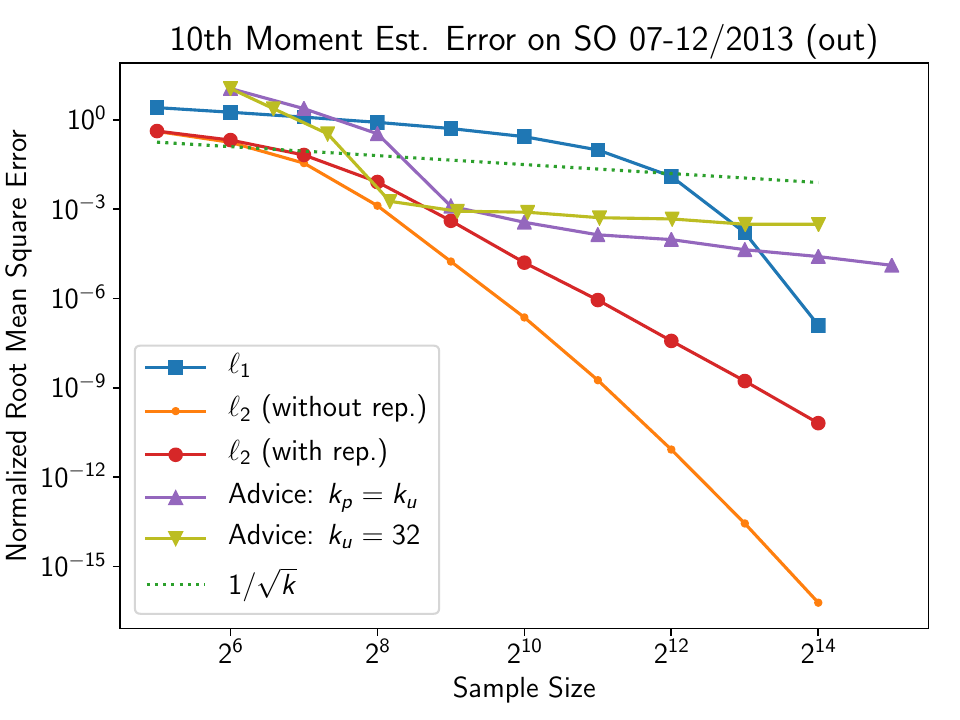}
\caption{NRMSE for estimating the 3rd, 7th, and 10th moments on the Stack Overflow dataset (incoming and outgoing edges), based on past frequencies.}
\label{advicemore2:fig}
\end{figure*}
Additional results on the performance of sampling by advice are provided in Figures~\ref{advicemore:fig} and~\ref{advicemore2:fig}.

Estimates of the distributions of frequencies for all datasets are reported in Figure~\ref{moreestfreq:fig}.  The estimates are from PPSWOR ($\ell_1$ without replacement) and $\ell_2$ (with replacement) samples of sizes $k=32$ and $k=1024$.  For each key $x$ in the sample, we estimate its rank in the dataset, that is, the number of keys $y$ with frequency $w_y \geq w_x$.  The rank estimate for a sampled key $x$ is computed from the sample, by estimating the $f$-statistics for the threshold function $f(w) := I_{w\geq w_x}$.
The pairs of frequency and estimated rank (for each key in the sample) are then plotted.  The figures also provide the exact frequency distribution.
\begin{figure*}[t]
\centering
\includegraphics[width=0.49\textwidth]{dist_AOL.pdf}
\includegraphics[width=0.49\textwidth]{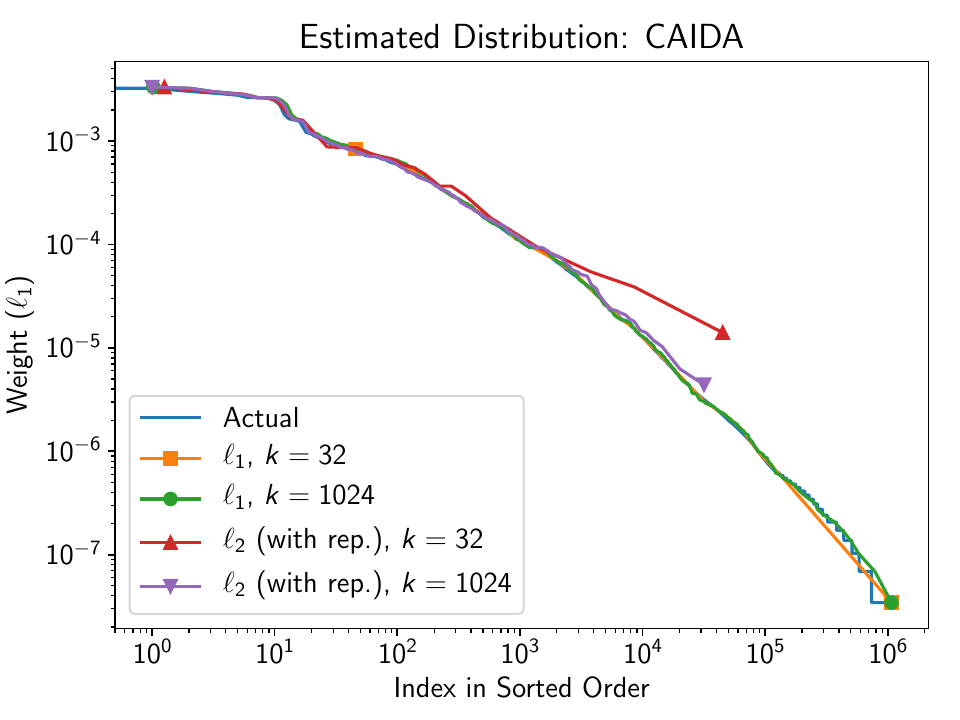}
\includegraphics[width=0.49\textwidth]{dist_UGR.pdf}
\includegraphics[width=0.49\textwidth]{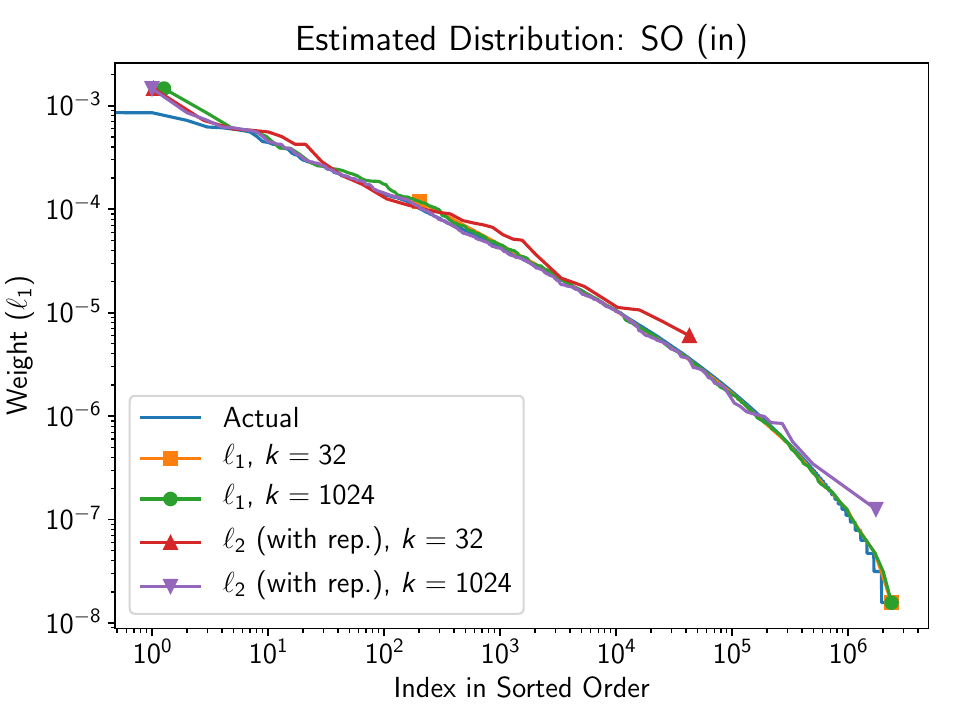}
\includegraphics[width=0.49\textwidth]{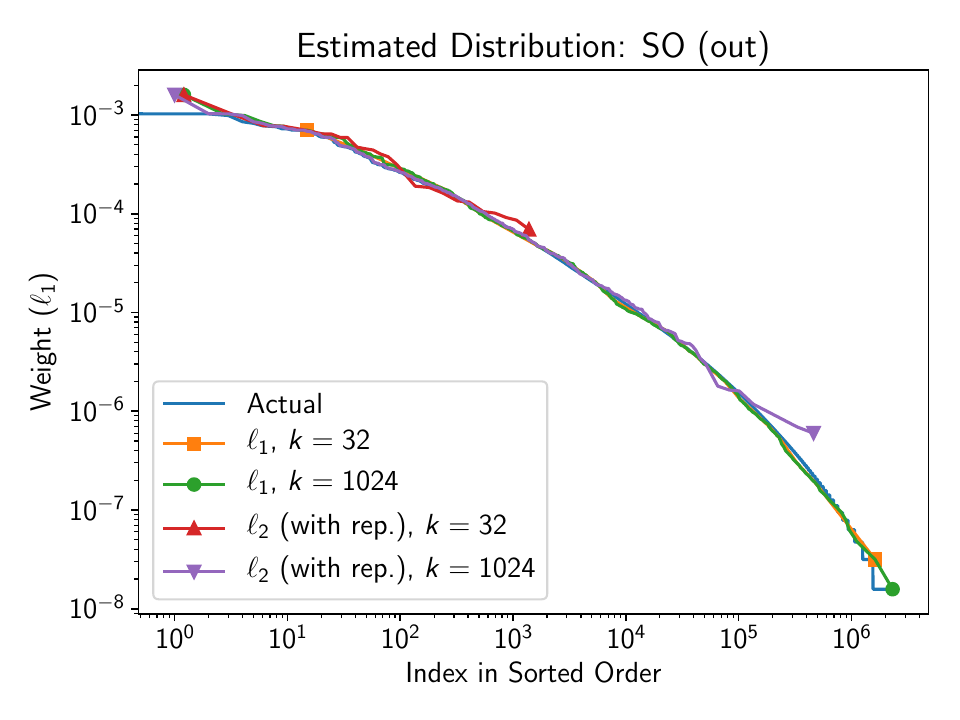}
\caption{Actual and estimated distribution of frequency by rank.  The estimates use PPSWOR and with-replacement $\ell_2$ sampling and sample sizes $k=32,1024$.}
\label{moreestfreq:fig}
\end{figure*}

Additional results on the estimation of moments from PPSWOR ($\ell_1$ sampling without replacement) and $\ell_2$ samples (with and without replacement) are reported in Figure~\ref{moremomentsest:fig}.
As suggested by our analysis, the estimates on all datasets are surprisingly accurate even with respect to the benchmark upper bound $1/\sqrt{k}$ (which follows from the variance bound of weighted with-replacement sampling tailored to the moment we are estimating).  The figures also show the advantage of without-replacement sampling on these skewed datasets.
\begin{figure*}[thb]
\centering
  \includegraphics[width=0.32\textwidth]{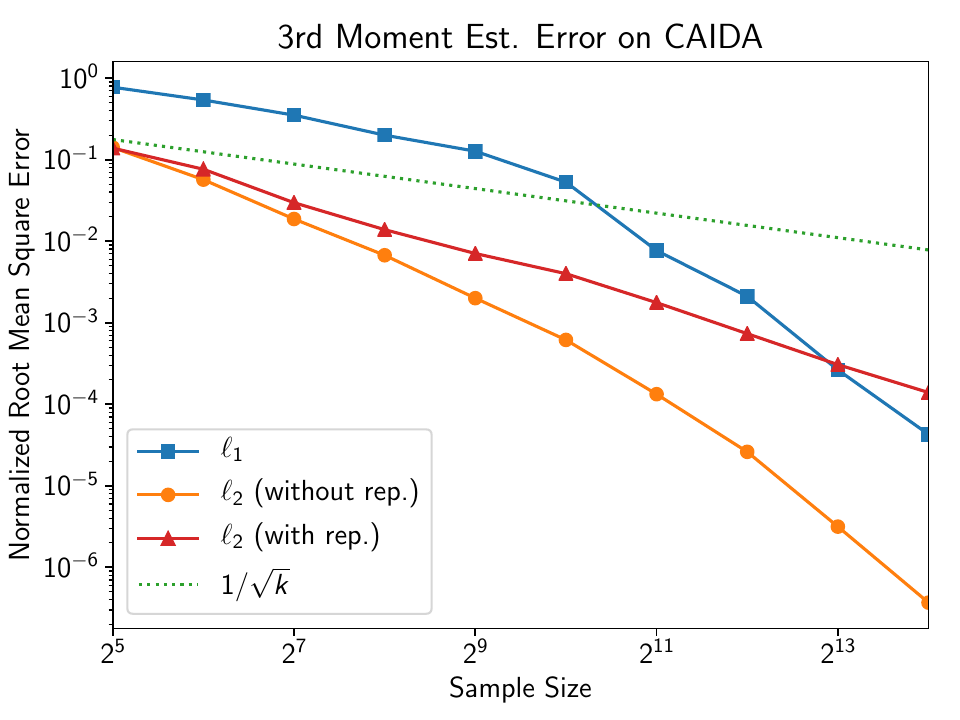}
  \includegraphics[width=0.32\textwidth]{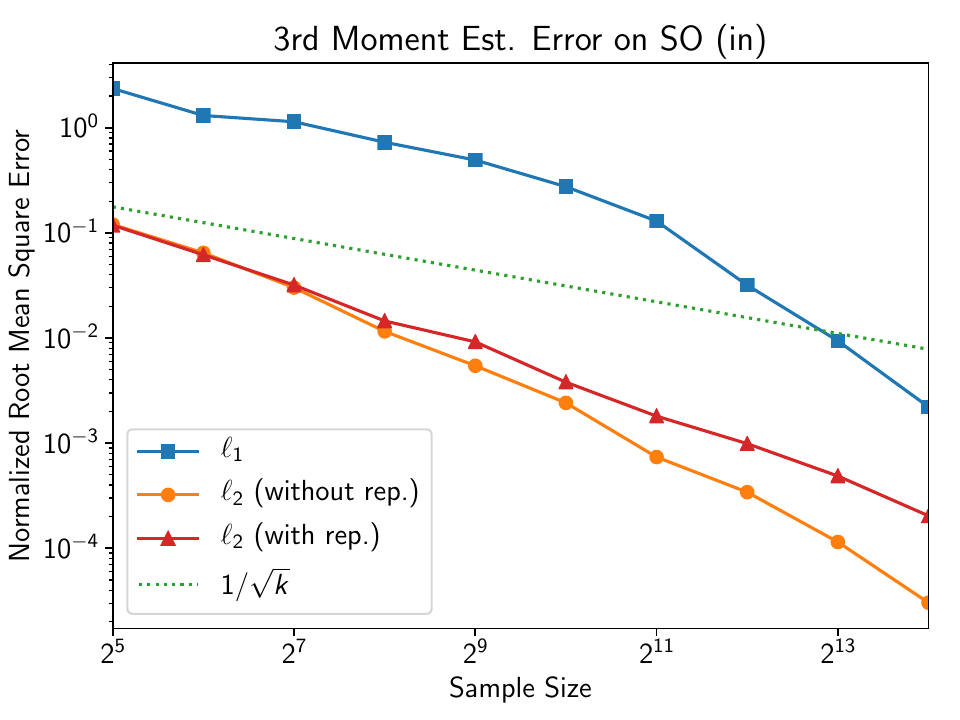}
  \includegraphics[width=0.32\textwidth]{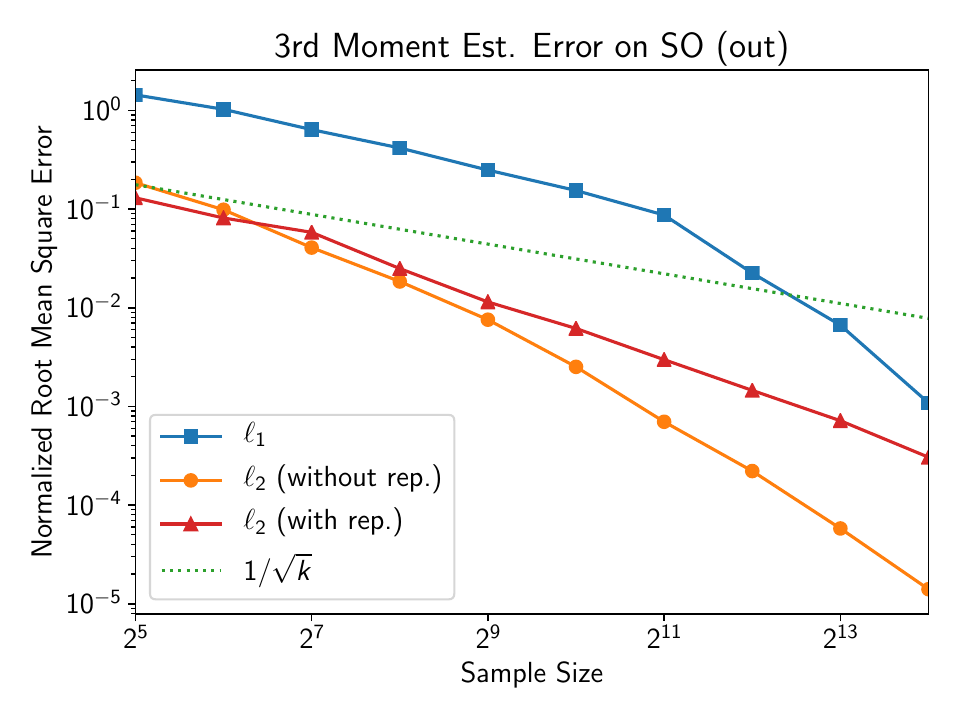}\\
    \includegraphics[width=0.32\textwidth]{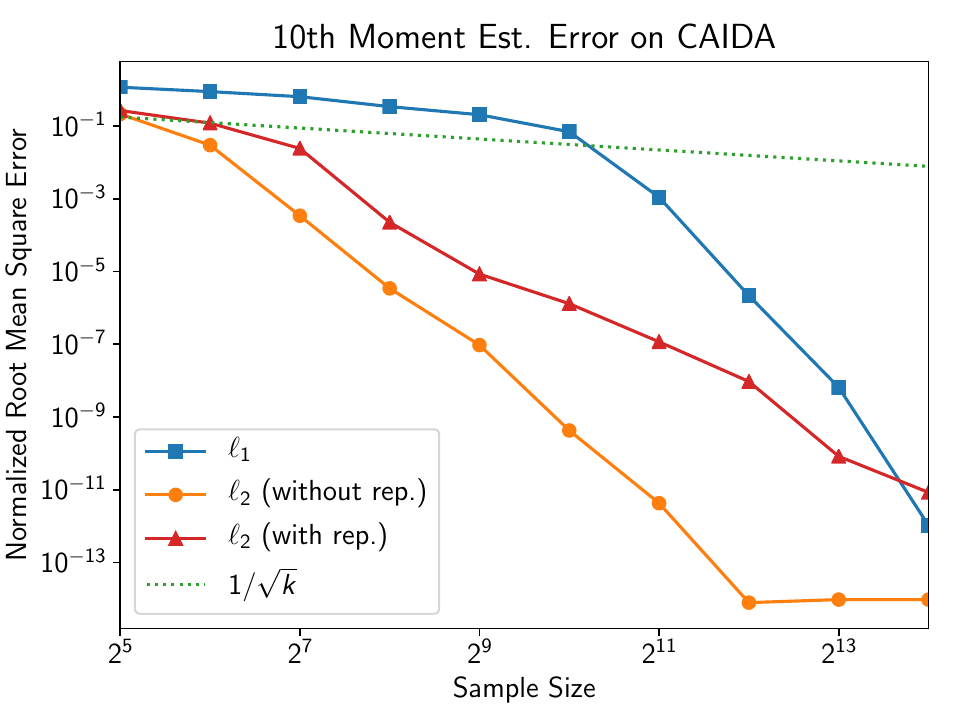}
    \includegraphics[width=0.32\textwidth]{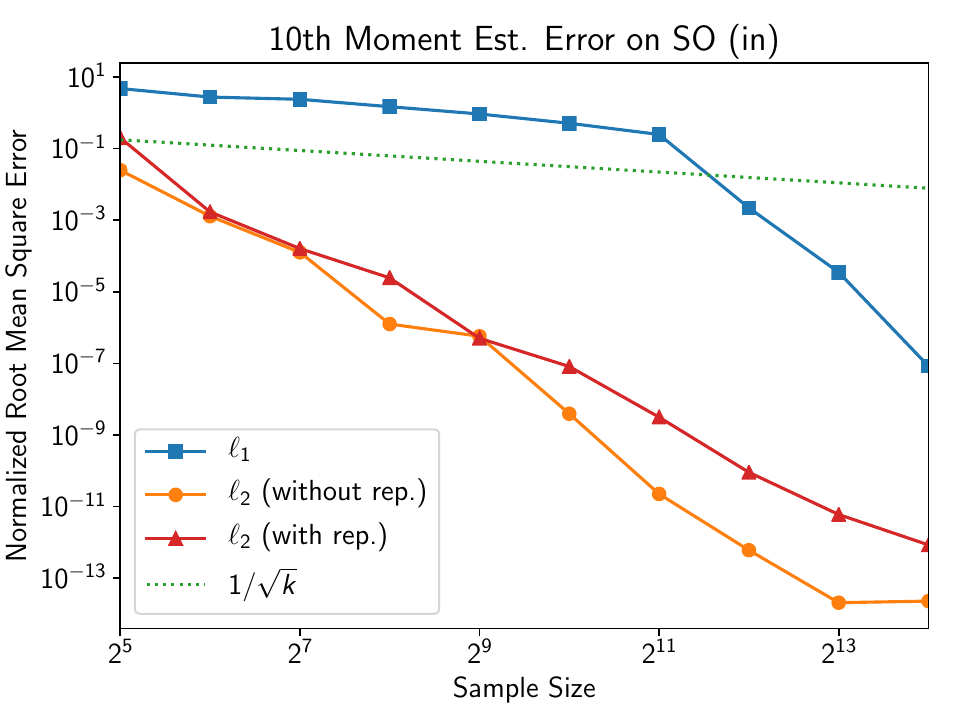}
    \includegraphics[width=0.32\textwidth]{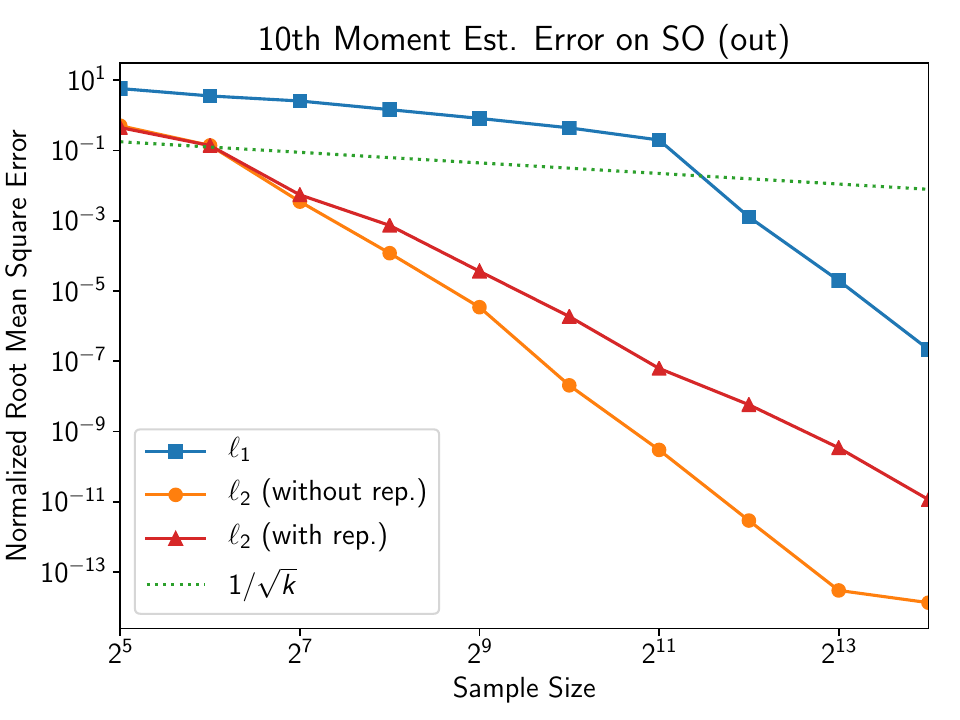}\\
    \includegraphics[width=0.32\textwidth]{aol_aggregated_all_ell_3.pdf}
    \includegraphics[width=0.32\textwidth]{ugr_may_week2_aggregated_directed_ell_3.pdf}\\
    \includegraphics[width=0.32\textwidth]{aol_aggregated_all_ell_10.pdf}
    \includegraphics[width=0.32\textwidth]{ugr_may_week2_aggregated_directed_ell_10.pdf}
\caption{Estimating the 3rd and 10th moments on various datasets using PPSWOR ($\ell_1$ sampling without replacement) and $\ell_2$ samples (with and without replacement).  The error is averaged over 50 repetition.}
\label{moremomentsest:fig}
\end{figure*}

\end{document}